\documentclass[11pt]{article}
\usepackage[utf8]{inputenc}
\usepackage{evan}
\usepackage{indentfirst}
\usepackage{color}
\usepackage{authblk}
\usepackage{textcomp}

\usepackage{algorithm, algpseudocode}

\newcommand{\E}{\operatorname{\mathbb{E}}}
\newcommand{\Lap}{\operatorname{Lap}}
\newcommand{\C}{\mathbb{C}}
\newcommand{\R}{\mathbb{R}}
\newcommand{\N}{\mathbb{N}}
\newcommand{\V}{\operatorname{\mathbb{V}}}

\usepackage{lineno}

\begin{document}

\title{Debiasing Functions of Private Statistics in Postprocessing}

\author[1]{Flavio Calmon}
\author[1]{Elbert Du}
\author[1]{Cynthia Dwork}
\author[2]{Brian Finley}
\author[3]{Grigory Franguridi}

\affil[1]{Harvard University\thanks{This work was supported in part by Simons Foundation Grant 733782 and Cooperative Agreement CB20ADR0160001 with the United States Census Bureau. This 
material is also based upon work supported by the National Science Foundation under Grant No. CIF-2312667 and CIF-2231707.}}
\affil[2]{U.S. Census Bureau\thanks{Works (articles, reports, speeches, software, etc.) created by U.S. Government employees are not subject to copyright in the United States, pursuant to 17 U.S.C. \S 105.  International copyright, 2024, U.S. Department of Commerce, U.S. Government. Any opinions and conclusions expressed herein are those of the authors and do not reflect the views of the U.S. Census Bureau.}
}
\affil[3]{USC Dornsife Center for Economic and Social Research\thanks{\textcopyright{} Flavio Calmon, Elbert Du, Cynthia Dwork, Grigory Franguridi -- Creative Commons (CC BY-SA 4.0)}
}

\maketitle

\begin{abstract}

Given a differentially private unbiased estimate $\tilde{q}=q(D) +\nu$ of a statistic $q(D)$, we wish to obtain unbiased estimates of \textit{functions} of $q(D)$, such as $1/q(D)$, solely through post-processing of $\tilde{q}$, with no further access to the confidential dataset~$D$.  To this end, we adapt the deconvolution method used for unbiased estimation in the statistical literature, deriving unbiased estimators for a broad family of twice-differentiable functions when the privacy-preserving noise $\nu$ is drawn from the Laplace distribution (Dwork \textit{et al.}, 2006).
    We further extend this technique to a more general class of functions, deriving approximately optimal estimators that are unbiased for values in a user-specified interval (possibly extending to $\pm \infty$).

    We use these results to derive an unbiased estimator for private means when the size $n$ of the dataset is not publicly known. In a numerical application, 
    we find that a mechanism that uses our estimator to return an unbiased sample size and mean outperforms a mechanism that instead uses the previously known unbiased privacy mechanism for such means (Kamath \textit{et al.}, 2023).
    We also apply our estimators to develop unbiased transformation mechanisms for per-record differential privacy, a privacy concept in which the privacy guarantee is a public function of a record's value (Seeman \textit{et al.}, 2024).  Our mechanisms provide stronger privacy guarantees than those in prior work (Finley \textit{et al.}, 2024) by using Laplace, rather than Gaussian, noise.

    Finally, using a different approach, we go beyond Laplace noise by deriving unbiased estimators for polynomials under the weak condition that the noise distribution has sufficiently many moments.
    
\end{abstract}
\pagebreak
\tableofcontents
\pagebreak
\section{Introduction}

Differential privacy (DP) has become widely accepted as a ``gold standard'' of privacy protection in statistical analysis. In particular, it has been adopted by many companies such as Google, Meta, and Apple to protect customer data and by the U.S. Census Bureau to protect respondent data in the 2020 Census \cite{abowd20222020censusdisclosureavoidance}. DP mechanisms work by adding randomness to all statistics published from a protected database. This added noise guarantees that no attacker can confidently determine from the published statistics whether a particular record is included in the dataset, thereby preserving its privacy. 
    
Among DP mechanisms, \textit{additive mechanisms} are canonical and widely used. These simply add data-independent, zero-mean random noise to the statistics. For example, the Laplace mechanism adds Laplace-distributed noise and is one of the first and most fundamental DP mechanisms \cite{Dwork_McSherry_Nissim_Smith_2017}. The scale of the added noise needs to be proportional to the statistics' global sensitivity -- the greatest amount by which the statistic could change upon the addition or deletion of a single record. Intuitively, this ensures that there is enough noise to mask the presence or absence of any particular record.

Often, however, the statistics to which noise is added differ from the final statistics of interest. In these cases, the statistics of interest must be estimated using the available noisy statistics. Suppose that the noisy statistic $\tilde{q}$ is formed by applying an additive mechanism to the univariate statistic $q$, but that we want to learn $f(q)$, not $q$. Even though $\tilde{q}$ is unbiased for $q$, the plug-in estimator $f(\tilde{q})$ is not generally unbiased for $f(q)$. When unbiasedness is desired, other estimators must be used.

To address this problem, we first derive unbiased estimators for the Laplace mechanism, for a general class of twice-differentiable functions $f$ -- those which are also \textit{tempered distributions} (Section~\ref{sec:unbiasLaplace}). This builds on \cite{hillebrand2023unbiased}, which develops recursive algorithms that are unbiased estimators for polynomials in Laplace variables. Our paper provides estimators for a large class of non-polynomial functions and gives a simple, closed-form estimator for polynomials.  We also provide methods to adapt functions that are not tempered distributions in a way that permits unbiased estimation over a subset of $q$'s domain (Section~\ref{sec:unbiasLaplaceNonTempDist}). This extension lets us provide unbiased estimators for the case when $f(q) = 1/q$, which, in turn, lets us provide unbiased estimators of ratio statistics. Such cases arise frequently in practice, as discussed below. Finally, we derive unbiased estimators for a very general class of additive mechanisms when $f$ is a polynomial (see Section~\ref{sec:unbiasPoly}).

There are several reasons why noise may not be added directly to the statistic of interest and bias in the plug-in statistic must be considered. A leading case occurs when $f(q)$ has a much higher global sensitivity than $q$. For example, when the domain of $q$ includes 0 or values arbitrarily close to 0, the global sensitivity of $f(q) = 1/q$ is typically infinite and no amount of noise provides a finite DP guarantee. This same problem affects the many statistics which can be expressed as ratios of low-sensitivity statistics. For example, the mean is the ratio of a sum and a sample size. Likewise, in a simple linear regression of the regressand $y$ on the regressor $x$, the ordinary least squares (OLS) estimator of the slope coefficient is the ratio of the empirical (co)variances $\text{Cov}[x,y]$ and $\text{Var}[x]$. It is common, then, to add noise to the low-sensitivity statistics that form these ratios and use the plug-in estimator for the ratio statistic of interest. See, for example, \cite{alabi2020differentially} for such a treatment of the OLS estimator.

Noise may also be added to statistics that are not of direct interest because data curators, such as government agencies, may publish noisy microdata or a noisy predetermined set of aggregates for general-purpose use. For example, a researcher trying to learn the proportion of the population with doctoral degrees may only have access to published noisy totals of the general population and the population of degree holders. The plug-in estimator of the mean is, as above, the ratio of these noisy totals.
    
This situation may also arise because, under DP, there is a limited ``privacy budget" which is drawn upon every time we use the raw data to release another (noisy) statistic. Splitting the budget among more statistics requires that more noise be added to each of them. This makes it beneficial to ``re-use'' statistics whenever possible. For example, in Section~\ref{sec:Numerical_Means}, we develop a DP mechanism that uses our results to provide private unbiased estimates of a mean and sample size. This mechanism obtains a noisy sample size via the Laplace mechanism and then re-uses it to estimate the denominator of the mean statistic. We find that this approach outperforms an alternative mechanism that uses an existing method from \cite{KamathEtAl2023Trilemma} to add noise directly to the mean query, without re-using the noisy sample size.

Again, these scenarios all have in common that the plug-in estimator $f(\tilde{q})$ will typically be biased.\footnote{In fact, in the case with the Laplace mechanism and $f(q)=1/q$, the expectation and all higher moments of the plug-in estimator $f(\tilde{q})$ fail to even exist, implying that the estimator has extremely fat tails and is very prone to returning extreme outliers. This affects the ratio statistics discussed above, as well. The unbiased estimator we develop for this case in Section~\ref{subsec:UnbiasEstForInverse} possesses finite moments of all orders.} To see why unbiasedness is desirable, recall that the bias and variance of a sum of $n$ uncorrelated estimates are respectively the sums of the estimates' biases and variances. Accordingly, the sum's bias increases at the rate $O(n)$ while its standard deviation grows at the rate $O(\sqrt{n})$. The sum's overall RMSE therefore grows at the $O(\sqrt{n})$ rate if the estimates are unbiased, but at the faster $O(n)$ rate otherwise.

For example, consider the following simple example: suppose the true value of some quantity of interest is $1$, but each time we try to learn the value of this quantity, we get a fresh draw from the distribution $\mathcal{N}(1,100)$. A mechanism that ignores the data and returns $0$ has bias $1$ and variance $0$, resulting in an overall RMSE of $1$. On the other hand, reporting the value of any single draw would have bias $0$ and variance $100$. On average, this will be off by around $10$. Thus, on individual draws, the first mechanism is more accurate. However, if we take the mean of $10,000$ such draws, the mechanism that always returned $0$ still gives a mean of $0$, which is still off by $1$. On the other hand, the mean of the $10,000$ draws now has variance $\frac{100}{10,000} = \frac{1}{100}$, resulting in an RMSE of $\frac{1}{10}$. That is, we would now expect this estimate to be off by around $0.1$, which is a significant improvement over the biased estimator.

This makes unbiasedness very important in meta-analyses, which aggregate multiple estimates. It is also important when adding noise to a large number of quantities whose sums are of independent interest. This situation commonly arises in the local model of DP, where an extra layer of privacy is obtained by adding noise to every record even before entrusting it to the data curator. Sums or means using these noisy records could then be subject to severe error if the record-level estimates being summed are biased. For example, with network data, the count of $k$-stars (i.e., sets of $k$ edges sharing a node) is a sum of polynomials in each node's degree. In experiments with network data protected by local DP, \cite{hillebrand2023unbiased} find that a mechanism that simply sums unbiased estimates of these polynomials outperforms the L2 error of prior work by factors as high as 5 orders of magnitude.

Likewise, unbiasedness is key when noise is added to relatively disaggregate sums with the expectation that they can be aggregated further to obtain sums for larger groups. For example, \cite{seemanEtAl2024} and \cite{finleyEtAl2024} develop mechanisms for use with per-record DP -- a variant of DP whose privacy guarantees differ between records, and which is being considered by the Census Bureau for use with its County Business Patterns (CBP) data product \cite{BeckomEtAl2023}. \cite{finleyEtAl2024} develop \textit{transformation mechanisms} for this purpose, which improve privacy guarantees by adding noise to concave functions of $q$ rather than to $q$ itself. Estimates of $q$ must then be obtained from these noisy transformed values. The CBP data includes sums of employment and payroll, grouped by finely divided geographies and industry codes. If these transformation mechanisms were used for the CBP and the estimator of $q$ for these sums were biased, further aggregates of these estimates to obtain, say, state-level sums would be subject to severe biases.

In Section~\ref{sec:PRDP}, we apply our estimators to create variants of these transformation mechanisms that satisfy a stronger type of per-record DP guarantee than the ones originally proposed in \cite{finleyEtAl2024}.

In this paper, we make the following contributions:
\begin{enumerate}
    \item We derive closed-form unbiased estimators for a large class of functions -- twice-differentiable functions that are tempered distributions -- when the Laplace mechanism is used. We also develop estimators that are unbiased for subsets of the statistic's domain for functions that are not in this class.
    \item We exposit the deconvolution method from the statistics literature (e.g., \cite{VoinovNikulin1993}, page 185) for deriving unbiased estimators. This could be used to derive estimators for further functions and further mechanisms, and we believe its use in DP is novel and of independent interest.
    \item We apply our unbiased estimators to create novel unbiased privacy mechanisms for per-record DP, a new variant of DP being considered for use by the Census Bureau \cite{BeckomEtAl2023}.
    \item We derive closed-form unbiased estimators for polynomial functions of statistics privatized using any of a large class of additive mechanisms.
\end{enumerate}

\section{Differential Privacy, Unbiasedness, and Deconvolution}
\label{sec:prelim}
The following is the definition of differential privacy, introduced in \cite{Dwork_McSherry_Nissim_Smith_2017}:
\begin{definition}
Datasets $D$ and $D'$ are \emph{neighboring databases} if they differ by the inclusion of at most $1$ element.
\end{definition}

\begin{definition} \label{def:DP}
A mechanism $\mathcal{M}$ is $(\epsilon,\delta)$-differentially private ($(\epsilon,\delta)$-DP) if, for any pair of neighboring datasets $D,D'$ and any measurable set of possible outcomes $S$, we have

$$\Pr[\mathcal{M}(D) \in S] < e^\epsilon \cdot \Pr[\mathcal{M}(D') \in S] + \delta.$$
\end{definition}

Most of our work uses Fourier transforms \cite{Fourier_2009}. The following definitions and theorems are adapted from the textbook treatment in \cite{vanDijk+2013}.

\begin{definition} \label{def:FourierAbsInt}
The Fourier transform of an absolutely integrable function $f$ is
    $$F[f(x)](y) = \int_{-\infty}^\infty e^{-2\pi i yx} f(x) dx.$$

We often denote the Fourier transform of $f$ by $\hat f(y)$.
    The Fourier transform also has an inverse:
    $$F^{-1}[\hat{f}(y)](x) = \int_{-\infty}^\infty e^{2\pi i yx} \hat{f}(y) dy.$$
\end{definition}

There are many important functions we wish to compute unbiased estimates of which are not absolutely integrable. In particular, polynomials and the function $f(q) = 1/q$ are not absolutely integrable, so we must define the Fourier transform over a more general family of functions, \textit{tempered distributions}.\footnote{Technically, $1/q$ is not a tempered distribution either, but only because it is poorly behaved at 0. This will be addressed in Section~\ref{sec:unbiasLaplaceNonTempDist}.} Importantly, in this use, the term ``distribution'' does not refer to a probability distribution. Rather, it refers to a class of objects also known as ``generalized functions.'' For the purposes of this work, we can largely restrict ourselves to working with tempered distributions which are also functions, though there exist tempered distributions which are not functions, such as the Dirac delta ``function''. Below, we specialize the relevant theory to the case of tempered distributions that are also functions, but the interested reader should see Appendix~\ref{app:Prelim} for the more general case.

For our purposes, then, tempered distributions can be thought of as functions that may not be absolutely integrable, but which grow no faster than a polynomial. Formally, this is expressed by the condition that the product of a tempered distribution and any function in the Schwartz space (defined below) is integrable.

\begin{definition} \label{def:SchwartzSpace}
    The Schwartz space $S(\RR)$ is defined as follows:
    $$S(\RR) = \{s:\RR \rightarrow \C \mid s \in C^\infty, \,\, \sup_{x \in \RR} |x^m s^{(n)}(x)| < \infty \quad \forall m,n \in \N \},$$
    where $\N$ denotes the set of non-negative integers and $s^{(n)}$ denotes the $n^\text{th}$ derivative of $s$.
    That is, functions in $S(\RR)$ are infinitely differentiable everywhere and they -- along with all of their derivatives -- go to 0 at a super-polynomial rate.
\end{definition}

Note that all the functions $s \in S(\RR)$ are absolutely integrable, so their Fourier transforms $\hat s$ exist. With the Schwartz space so defined, we introduce tempered distributions below.

\begin{definition} \label{def:tmprDist}
    A function $f$ is a tempered distribution if and only if, for all $s \in S(\R)$,
    $$\int_{-\infty}^\infty f(x) s(x)dx \in \C.$$
\end{definition}

Definition~\ref{def:FourierAbsInt} introduces the Fourier transform only for absolutely integrable functions. The following definition extends it to all tempered distributions.\footnote{To see that Definition~\ref{def:FourierAbsInt} implies Definition~\ref{def:DistFourier} for absolutely integrable functions, note that the $e^{-2\pi ixy}$ term does not depend on the function $f$, so swapping the order of integration with Fubini's theorem immediately gives us the equality in Definition~\ref{def:DistFourier}.}

\begin{definition} \label{def:DistFourier}
    When it exists, $\hat{f}$ is the function such that for all $s \in S(\R)$,
    $$\int_{-\infty}^\infty \hat{f}(x) s(x)dx = \int_{-\infty}^\infty f(x) \hat{s}(x)dx.$$
\end{definition}

Technically, the Fourier transform of a tempered distribution always exists and is a tempered distribution, but may not also be a function, even when the distribution being Fourier-transformed is a function. See Appendix~\ref{app:Prelim} for details.

The deconvolution method we use to derive unbiased estimators in Section~\ref{sec:unbiasLaplace} is applicable because, as explained below, the requirement that an estimator be unbiased can be expressed in terms of a convolution.

\begin{definition}
    The convolution of functions $f$ and $g$ is
    $$(f*g)(x) = \int_{-\infty}^\infty f(z) g(x-z) dz.$$
\end{definition}

Critically, the Fourier transform of a convolution is the product of the convolved functions' Fourier transforms.

\begin{theorem}
\label{thm:convolution}
    (\cite{vanDijk+2013} section 7.1 property c)
    $$\widehat{(f*g)}(y) = \hat{f}(y) \cdot \hat{g}(y).$$
\end{theorem}

We will also need the following theorem to derive unbiased estimators for the case of Laplace noise.


\begin{theorem}
    \label{thm:tmpDistFourierDeriv}
    (\cite{vanDijk+2013} section 7.8 Example 5)
    For any tempered distribution $f$, the Fourier transform of its $k^{th}$ derivative $f^{(k)}$ is $\widehat{f^{(k)}} = (2\pi i y)^k \hat{f}$.
\end{theorem}


With the query $q$ and its privacy-preserving noisy estimate $\tilde{q}$, we say that an estimator $g$
is unbiased for $f(q)$ if \begin{align} \label{eq:unbias}
f(q) = \E[g(\tilde{q}) | q].
\end{align}
By conditioning on the true query value, $q$, we treat the database as fixed. Our estimators, then, are unbiased with respect to the randomness in the noise being added for privacy. Throughout the rest of this paper, all expectations are conditional on $q$ unless otherwise noted and we suppress the extra conditioning notation so that $\E[.] \equiv \E[.|q]$.

Let the noise added for privacy be independent of the database and denote its PDF by $r$. The deconvolution method, as seen, for example, on page 185 of \cite{VoinovNikulin1993}, starts by noting that if $g$ is unbiased for $f(q)$, then Equation~\ref{eq:unbias} can be reexpressed in terms of a convolution:
\begin{align} \label{eq:unbiasConv}
    f(q) = \E[g(\tilde{q})] = \int_{-\infty}^\infty g(\tilde{q}) r(\tilde{q} - q) d\tilde{q} = (g*r)(q).
\end{align}
With the unbiasedness equation in this form, Theorem~\ref{thm:convolution} lets us Fourier-transform both sides to turn the convolution on the right-hand side into a simple multiplication. Finally, we simply solve for the Fourier transform of $g$ in terms of the Fourier transforms of $f$ and $r$ and inverse-Fourier-transform the result. Formally,
\begin{align}
    f(q) = (g*r)(q) \iff \hat{f}(y) = \hat{g}(y) \hat{r}(y) \iff \hat{g}(y) = \frac{\hat{f}(y)}{\hat{r}(y)} \iff g(x) = F^{-1}\left[\frac{\hat{f}(y)}{\hat{r}(y)}\right](x),
\end{align}
assuming the existence of all the involved Fourier and inverse Fourier transforms.

\section{Unbiased Estimation with Laplace Noise}
\label{sec:unbiasLaplace}

A standard mechanism for differential privacy perturbs the query with Laplace noise scaled to the global sensitivity of a query, which is the maximum difference between the query values on neighboring databases. That is, to achieve $(\epsilon,0)$-DP when releasing the value of a query $q$ with global sensitivity $\Delta$, we can simply release $q + \Lap(0,\frac{\Delta}{\epsilon})$ \cite{Dwork_McSherry_Nissim_Smith_2017}.

Our primary contribution is deriving unbiased estimators for functions of $q$ when we only have access to the value of $q + \Lap(0,b)$, for some noise scale parameter $b$. These estimators are unique (up to their values on a set of measure zero).

\begin{theorem}
\label{thm:unbiased}
    \item Let $\tilde{q} \sim q + \Lap(0,b)$. 
    \begin{itemize}
        \item For any twice-differentiable function $f:\R \to \R$ that is a tempered distribution, $f(\tilde{q}) - b^2 f''(\tilde{q})$ is an unbiased estimator for $f(q)$.
        \item For any function $f:\R \to \R$, if two estimators $g_1(\tilde{q})$ and $g_2(\tilde{q})$ are unbiased for $f(q)$, then $g_1$ and $g_2$ are equal almost everywhere.
    \end{itemize}
\end{theorem}

\begin{proof}
    See Appendix~\ref{App: section3}.
\end{proof}

Some examples of unbiased estimators are given below.

\begin{example}
\label{ex:UnbiasedEstExamples_powerFunc}
\begin{enumerate}
    \item Any power function $f(q) = q^k$ has unbiased estimator $\tilde{q}^k - b^2 k(k-1) \tilde{q}^{k-2}$. In particular, for $f(q) = cq$ for any constant $c$, the unbiased estimator is also $c\tilde{q}$. Section 3.1 of \cite{hillebrand2023unbiased} derives this estimator in the form of a recursive algorithm. We contribute the closed form here to simplify computation and facilitate intuitive understanding.
    \item Within the set of twice differentiable functions $f$ that are tempered distributions, Theorem~\ref{thm:unbiased} allows us to determine which functions are unbiased estimators of themselves. When $f(\tilde{q})$ is unbiased for $f(q)$, we have $\E[f(\tilde{q})] = \E[f(\tilde{q}) - b^2 f''(\tilde{q})] \implies \E[f''(\tilde{q})] = 0$. By the second part of Theorem ~\ref{thm:unbiased} and the unbiasedness of the zero function for zero, this implies $f''(x) = 0$ almost everywhere, so $f(x)$ must be linear. The naive plug-in estimator, then, is biased for any nonlinear function in this class. This highlights the usefulness of Theorem~\ref{thm:unbiased}. 

    We can similarly characterize the $f$ whose unbiased estimators are simply linear transformations of the plug-in estimator -- that is, $f$ for which $\E[\alpha f(\tilde{q}) + \beta] = f(q)$ for some $\alpha, \beta \in \R$. By Theorem~\ref{thm:unbiased}, these functions satisfy $\E[\alpha f(\tilde{q}) + \beta] = \E[f(\tilde{q}) - b^2 f''(\tilde{q})]$, and so satisfy $\alpha f(x) + \beta = f(x) - b^2 f''(x)$ almost everywhere. When $\alpha \neq 1$, solutions to this differential equation take the form\footnote{The case where $\alpha = 1$ and $\beta = 0$ is dealt with above.}
    $$f(x) = c_1 e^{\frac{\sqrt{1-\alpha}}{b} x}+c_2 e^{-\frac{\sqrt{1-\alpha }}{b}x} + \frac{\beta}{1-\alpha}.$$

    Using Euler's formula, we can see that tempered distributions in this class include functions of the form $f(x) = c \cos(u x)$ and $f(x) = c\sin(u x)$. Nonetheless, this is still a rather restricted class of functions.

\end{enumerate}
\end{example}

\begin{remark} \label{rem:fNot2Diff}
When the function $f$ is not twice differentiable but is a tempered distribution, an analog of Theorem~\ref{thm:unbiased} holds. This relies on the use of an alternative notion of the derivative that applies to all tempered distributions -- the \textit{distributional derivative}. For background on this derivative concept, see Appendix~\ref{app:Prelim}.

In this case, we still have $\E[f(\tilde{q})] - \E[b^2 f''(\tilde{q})] = f(q)$, but the distributional derivative $f''$ is a tempered distribution which is not a function. This does not give us an unbiased estimator, but instead we can rearrange to obtain the bias, as a function of $q$, of the naive plug-in estimator $f(\tilde{q})$:

\begin{align} \label{eq:LapEstBias}
    \E[f(\tilde{q})] - f(q) = \E[b^2 f''(\tilde{q})] = \frac{b}{2} \int_{-\infty}^\infty f''(q + x)  e^{-|x|/b} dx.
    \end{align}

For example, let $f(x) = |x|$, $f'(x) = -1$ for $x \le 0$ and $1$ for $x > 0$ (the discontinuity at 0 is irrelevant since $\{0\}$ is a set of measure 0). Then $f''(x) = 2 \delta(x)$, where $\delta$ is the Dirac delta function (defined in Example~\ref{ex:dirac} in Appendix~\ref{app:Prelim}). The bias is simply $\frac{b}{2} \cdot 2 e^{-|q|/b} = b e^{-|q|/b}.$

Whether or not $f$ is twice differentiable, Equation~\ref{eq:LapEstBias} suggests the intuition that the plug-in estimator will have greater bias when $f$ has greater curvature near the true query value $q$.
\end{remark}

\section{Extension to Functions that are not Tempered Distributions}
\label{sec:unbiasLaplaceNonTempDist}

If the function $q \mapsto f(q)$ is not a tempered distribution, we can often bound the domain such that it is continuous and twice-differentiable in that domain. That is, suppose we know a priori that $q \ge L$ for some lower bound $L \in \RR$. This is often the case in differential privacy, as the global sensitivity of a sum query is finite only if the data's domain is bounded. Likewise, counts can often be lower bounded by $1$. Then, suppose we replace the function $f$ with some function
\begin{align}
    \label{eq:f-tilde-def} 
    \tilde{f}(q) = \left\{\begin{matrix} f(q) & q \ge L\\ h(q) & q < L\end{matrix}  \right.,
\end{align}
where $h(L) = f(L)$ and $h$ is twice differentiable with $h'(L) = f'(L)$ and $h''(L) = f''(L)$. The function $\tilde{f}$ is thus twice differentiable. Assuming that $h(q)$ and $f(q)$ and their derivatives grow no faster than a polynomial as, respectively, $q \to -\infty$ and $q \to \infty$, $\tilde{f}$ is a tempered distribution, as well. We can then apply Theorem~\ref{thm:unbiased} to get an unbiased estimator of $\tilde f$, i.e.
\begin{align}
    \mathbb{E}\left[\tilde{f}(\tilde{q}) - b^2 \tilde{f}''(\tilde{q})\right] = \tilde{f}(q).
\end{align}
With the assumption that $q \ge L$, we have $\tilde{f}(q) = f(q)$, making this estimator unbiased for $f(q)$, as well.

\begin{example}
\label{ex:taylor}
    For $f(q) = \frac{1}{q}$ and $L = 1$, we need to find some function $h$ such that $h(1) = 1, h'(1) = -1, h''(1) = 2$. An example of such a function is $h(q) = 1 - (q-1) + (q-1)^2.$ We can generically use polynomials for $h$ whenever $f$ grows at most polynomially as $q \rightarrow \infty$ and is twice differentiable for $q \geq L$.
\end{example}

We now focus on optimizing this method over polynomial extensions for a particular function of interest: $f(q) = 1/q$.

\subsection{Unbiased Estimation for $f(q) = 1/q$}
\label{subsec:UnbiasEstForInverse}

We have shown that it is possible to construct a function that permits unbiased estimation as long as it is twice differentiable on some domain that the true query value is known to be in, and, if this domain is unbounded, as long as the function does not grow too quickly.
In this section, we show how to optimally choose the function $h$ in the above construction.

We restrict ourselves to polynomial functions $h$ for two reasons.
First, the solution among polynomials of fixed degree is efficiently computable. Second, when the optimal function $h$ is a twice continuously differentiable tempered distribution, polynomials can approximate this function arbitrarily well, in the sense that the expected squared error of the polynomial-based estimator can be made arbitrarily close to optimal. This follows from Theorem~\ref{thm:polyApprox}.

\begin{theorem}[Polynomial approximation]
    \label{thm:polyApprox}
    Let $L \in \R$ and let $f:[L,\infty) \to \R$ be twice differentiable and a tempered distribution. Let $\mu$ be a probability measure such that the integrals $\int_L^\infty f(q) e^{(L-q)/b} d\mu(q)$ and $\int_L^\infty e^{(L-q)/b} d\mu(q)$ exist and are finite. With $w:(-\infty,L] \to \R$, let $\tilde{f}[w]$ denote the function
    \begin{align}
    \tilde{f}[w](q) = \left\{\begin{matrix} f(q) & q \ge L\\ w(q) & q < L\end{matrix}  \right.
    \end{align}
    
    Let $h:(-\infty,L] \to \R$ be an arbitrary twice continuously differentiable tempered distribution that satisfies $h(L) = f(L)$, $h'(L) = f'(L)$ and $h''(L) = f''(L)$. Denote the estimator $g \equiv \tilde{f}[h] - b^2 \tilde{f}[h]''$ and denote its expected squared error by 

    \begin{align}
        \alpha \equiv \int_{-\infty}^\infty \int_L^\infty \left(g(x) - f(q)\right)^2 \frac{1}{2b} e^{-|x-q|/b} d\mu(q) dx.
    \end{align}

    There exists a sequence of polynomials $(p_K)_{K=1}^\infty$ over $(\infty,L]$ that satisfy $p_K(L) = f(L)$, $p_K'(L) = f'(L)$ and $p_K''(L) = f''(L)$ such that the sequence of associated estimators $g_K \equiv \tilde{f}[p_K] - b^2 \tilde{f}[p_K]''$ satisfies

    \begin{align}
        \lim_{K \to \infty} \int_{-\infty}^\infty \int_L^\infty \left(g_K(x) - f(q)\right)^2 \frac{1}{2b} e^{-|x-q|/b} d\mu(q) dx = \alpha.
    \end{align}
\end{theorem}

See Appendix~\ref{App:section4} for a proof.

Now, letting $h$ be a polynomial, suppose our estimator is $g(x) = \tilde{f}(x) - b^2 \tilde{f}''(x)$ for $\tilde{f}$ defined in Equation~\ref{eq:f-tilde-def}. For our error metric, we consider the estimator's expected squared error, with the expectation taken over both the privacy noise and prior beliefs about $q$, reflected in the probability measure $\mu$. We define our estimator as the solution to the following constrained optimization problem:
\begin{align}&\min_{\tilde{f}} \int_{-\infty}^\infty \int_L^\infty \left(\tilde{f}(x)-b^2 \tilde{f}''(x) - f(q)\right)^2 \frac{1}{2b} e^{-|x-q|/b} d\mu(q) dx\\ \nonumber
&= \int_{L}^\infty \int_L^\infty \left(f(x) - b^2 f''(x) -f(q)\right)^2 \frac{1}{2b} e^{-|x-q|/b} d\mu(q) dx \\ \nonumber
&+ \min_g \int_{-\infty}^L \int_L^\infty \left(g(x)-f(q)\right)^2 \frac{1}{2b} e^{(x-q)/b} d\mu(q) dx \\ \nonumber
&\text{subject to $h(L) = f(L), h'(L) = f'(L)$, and $h''(L) = f''(L)$.}
\end{align}

Since the first double integral is constant with respect to $h$, optimizing this error metric is equivalent to optimizing
\begin{align}
    \min_g \int_{-\infty}^L \int_L^\infty \left(g(x)-f(q)\right)^2 \frac{1}{2b} e^{(x-q)/b} d\mu(q) dx,
\end{align}
subject to the same constraints.

For simplicity, we shall now treat $g$ as a function with domain $(-\infty,L]$, as that is the only region on which we are optimizing, so $g(x) = \sum_{i=0}^k a_i x^i$. There is a one-to-one correspondence between polynomials $g(x)$ and polynomials $h(x) = \sum_{i=0}^k b_i x^i$ where $a_i = b_i - b^2(i+2)(i+1)b_{i+2}$. Thus, we are considering extensions of $\tilde{f}(q)$ where the part to the left of the lower bound $L$ is a polynomial.

\begin{theorem}\label{thm:optimal}
    For any positive integer $k$, any real number $L \in \RR$, and any function $f$ which is twice differentiable on $[L,\infty)$, there is an algorithm that runs in time $poly(k)$ which computes the polynomial $g$ that minimizes 
    $$\int_{-\infty}^L \int_L^\infty \left(g(x)-f(q)\right)^2 \frac{1}{2b} e^{(x-q)/b} d\mu(q) dx$$
    over polynomials of degree $k$, satisfying the constraints $g(x) = h(x) - b^2 h''(x)$, $h(L) = f(L), h'(L) = f'(L)$, and $h''(L) = f''(L)$.
\end{theorem}

\begin{proof}
    See Appendix~\ref{app:thm15}.
\end{proof}

\begin{corollary}
    Provided that the optimal choice of $h$ is a twice continuously differentiable tempered distribution, there exists an efficient algorithm to approximate the optimal unbiased estimator of $f(q)$ given $q+Z$ for $Z \sim \Lap(0,b)$ and the prior knowledge that $q \ge L$.
\end{corollary}

This follows immediately from the fact that optimizing $g$ also optimizes $\tilde{f}(q)$.

Note that this result can be easily extended to the cases where we have only an upper bound or both an upper and lower bound. If we only have an upper bound, everything works out exactly the same as if we only have a lower bound. If we have both, suppose we know that $q \in [L,U]$ and define
\begin{align}
    \tilde{f}(q) = \left\{\begin{matrix} h_0(q) & q > U, \\
    f(q) & U \ge q \ge L,\\
    h_1(q) & q < L.
    \end{matrix}  \right.
\end{align}

Then the expected error (with the expectation over both the privacy noise and the prior on $q$) is
\begin{align}
  &\mathbb{E}_{q, x \sim q + Lap(0,b)} \left[ \left(\tilde{f}(x) - b^2 \tilde{f}''(x) - f(q)\right)^2 \right] \\
  &= \int_{-\infty}^\infty \int_L^U \left(\tilde{f}(x) - b^2 \tilde{f}''(x) - f(q)\right)^2 \frac{1}{2b} e^{-|x-q|/b} d\mu(q)dx. \nonumber
\end{align}

Just like before, the error incurred by $\tilde{f}(x)$ on $L \le x \le U$ is not affected by our choice of functions. Thus, we wish to compute
\begin{align}\min_{h,h'}&\int_{-\infty}^L \int_L^U \left(h_1(x) - b^2 h_1''(x) - f(q)\right)^2 \frac{1}{2b} e^{-|x-q|/b} d\mu(q)dx\\
+& \int_{U}^\infty \int_L^U \left(h_0(x) - b^2 h_0''(x) - f(q)\right)^2 \frac{1}{2b} e^{-|x-q|/b} d\mu(q)dx. \nonumber
\end{align}

Since there is no interaction between $h_0$ and $h_1$, we can minimize these integrals independently in the same way as above.

\section{Numerical Results: Application to Mean Queries}
\label{sec:Numerical_Means}

In this section, we illustrate the utility of our results by numerically comparing two mechanisms designed to return unbiased estimates of the sample size $n$ and the mean $m$ of an attribute $c \in [0,1]$ in the database $D$. Sample sizes are published alongside any reported means in most research applications, making this a realistic use case. One mechanism, $M_U$, returns an unbiased estimate of the mean using the results from Section~\ref{subsec:UnbiasEstForInverse}. The other, $M_{SS}$, uses the unbiased mean mechanism from \cite{KamathEtAl2023Trilemma} (see their Theorem D.6 and proof). To the best of our knowledge, this is the only published unbiased mechanism for means when the sample size is not treated as known. Both mechanisms use the Laplace mechanism with privacy budget $\epsilon_1$ to obtain the noisy sample size $\tilde{n}$. Each mechanism then allocates a separate privacy budget $\epsilon_2$ to obtain a noisy mean. Both mechanisms have a total privacy budget of $\epsilon_0 = \epsilon_1 + \epsilon_2$.

Denote attribute $c$ of record $r$ by $r.c$ and let $g(\tilde{q};k,L)$ be the unbiased estimator of $1/q$ from Section~\ref{subsec:UnbiasEstForInverse} with the generic query $q$ and polynomial extension of order $k$ for $q \leq L$. Algorithm~\ref{alg:OurMeanMech} lays out $M_U$. This algorithm applies the Laplace mechanism to the sum $s \equiv \sum_{r \in D} r.c$ and forms an unbiased estimate $\tilde{m}_U$ of the mean $m = s/n$ by multiplying the noisy sum $\tilde{s}$ by $g(\tilde{n};k,L)$. This is unbiased for $m$ as long as $n \geq L$.

For $n \geq L$, the variance of Algorithm~\ref{alg:OurMeanMech} is 
\begin{align}
    \V[\tilde{m}_U] &= (\E[\tilde{s}]^2 + \V[\tilde{s}])(\E[g(\tilde{n};k,L)]^2 + \V[g(\tilde{n};k,L)]) - \E[\tilde{s}]^2\E[g(\tilde{n};K,L)]^2 \\
    &= (s^2 + 2/\epsilon_1^2)(\E[1/n]^2 + \V[g(\tilde{n};k,L)]) - s^2/n^2.
\end{align}
For this section's numerical results, we calculate $\V[g(\tilde{n};k,L)]$ numerically. Finally, we note that it is straightforward to show that all moments of $\tilde{m}_U$ exist and are finite.

\begin{algorithm}

\caption{\label{alg:OurMeanMech} $M_U$}
\begin{algorithmic}[1]

    \Procedure{MeanPostprocess}{Database $D$, sample size privacy budget $\epsilon_1$, mean privacy budget $\epsilon_2$, polynomial order $k$, sample size lower bound $L$}
    \State $n \leftarrow \sum_{r \in D} 1$
    \State $\tilde{n} \leftarrow \Lap(n, 1/\epsilon_1)$ 
    \State $s \leftarrow \sum_{r \in D} r.c$
    \State $\tilde{s} \leftarrow \Lap(s, 1/\epsilon_2)$ 
    \State $\tilde{v} \leftarrow g(\tilde{n};k,L)$ \text{ (Unbiased estimator of $1/n$ for $n \geq L$)}
    \State $\tilde{m}_U \leftarrow \tilde{s}\tilde{v}$
    \State \textbf{output} $(\tilde{n}, \tilde{m}_U)$
    \EndProcedure
\end{algorithmic}    
\end{algorithm}

Let $T_3$ denote a random variable distributed according to a standard $t$ distribution with 3 degrees of freedom. $M_{SS}$ is laid out in Algorithm~\ref{alg:SSMeanMech}. This algorithm first forms a version $m_{SS}$ of the mean query that simply equals 1 if $n = 0$. It then scales the noise variable $T_3$ in proportion to an upper bound on the query's smooth sensitivity  \cite{KamathEtAl2023Trilemma}. The final noisy mean $\tilde{m}_{SS}$ is obtained by simply adding the scaled noise variable to $m_{SS}$.

The scaling factor for the noise is $\tau \max(e^{-\beta(n-1)}, 1/\max(n,1))$, where $\tau$ and $\beta$ satisfy $\epsilon_2 = 4\beta + 2/(\sqrt{3}\tau)$. The standard $t$ distribution with $d$ degrees of freedom has variance $d/(d-2)$ giving $\tilde{m}_{SS}$ a variance of 
\begin{align}
    \V[m_{SS}] = 3 \tau^2 \max(e^{-\beta(n-1)}, 1/\max(n,1))^2.
\end{align}
Because the $t$ distribution is symmetric, this mechanism is unbiased for $s/n$ as long as $n \geq 1$. Unlike $\tilde{m}_U$, however, the third and higher moments of $\tilde{m}_{SS}$ are infinite or do not exist. This is due to the $t$ distribution's very fat tails and implies that $\tilde{m}_{SS}$ is more liable than $\tilde{m}_U$ to produce extreme outliers.

\begin{algorithm}
\caption{\label{alg:SSMeanMech} $M_{SS}$}
\begin{algorithmic}[1]
    \Procedure{MeanSmoothSens}{Database $D$, sample size privacy budget $\epsilon_1$, mean privacy budget $\epsilon_2$, noise parameters $\beta,\tau$ satisfying $\epsilon_2 = 4\beta + 2/(\sqrt{3}\tau)$}
    \State $n \leftarrow \sum_{r \in D} 1$
    \State $\tilde{n} \leftarrow \Lap(n, 1/\epsilon_1)$
    \State $s \leftarrow \sum_{r \in D} r.c$
    \State $m_{SS} \leftarrow \textbf{ if $n \geq 1$: $s/n$ else: $1$}$
    \State $\tilde{m}_{SS} \leftarrow m_0 + T_3 \cdot \tau \max(e^{-\beta(n-1)}, 1/\max(n,1))$
    \State \textbf{output} $(\tilde{n}, \tilde{m}_{SS})$
    \EndProcedure
\end{algorithmic}    
\end{algorithm}

In our numerical evaluation of these mechanisms, we fix $\epsilon_1 = \epsilon_2 = m = .5$ for both mechanisms. For $M_{SS}$, we follow \cite{KamathEtAl2023Trilemma} in setting $\beta = \epsilon_2/12$ and $\tau = \sqrt{3}/\epsilon_2$. For $M_U$, we set $L = 1$ so that both mechanisms are unbiased for $n \geq 1$ and set $k=10$. 

With these settings, we compare the standard deviations (SDs) of the two mechanisms' mean estimates for a range of sample sizes. Because the mechanisms are unbiased, this is equivalent to their root mean squared error. The sample size estimates of both mechanisms are the same, so we do not report their properties.

Figures~\ref{fig:MnAppSDsBig} and \ref{fig:MnAppSDsZoom} plot the mechanisms' SDs as functions of $n$, with Figure~\ref{fig:MnAppSDsZoom} zooming in on larger values of $n$ for clarity. It is immediately clear that the $M_U$ has a larger SD than $M_{SS}$ for $n < 13$, and that this pattern reverses for larger $n$. For $n \leq 19$, however, both mechanisms have SDs greater than 1, making both unfit for most purposes at these sample sizes, given that the mean $m$ has the domain [0,1].

Figure~\ref{fig:MnAppSDsRel} shows the relative SD - that is, the ratio $SD(M_{SS})/SD(M_U)$ - as a function of $n$. For $n \geq 13$, the relative SD rises to a peak near 15 before settling down to an apparent constant of about 1.9 for $n \geq 115$.

Ultimately, $M_U$ appears to be the better mechanism for this setting; for any sample size where either mechanism returns useful results, $M_U$ has a substantially lower SD. The thinner tails of $M_U$ also recommend it as the better choice.

\begin{figure}[H]
    \centering
    \includegraphics[width=0.5\linewidth]{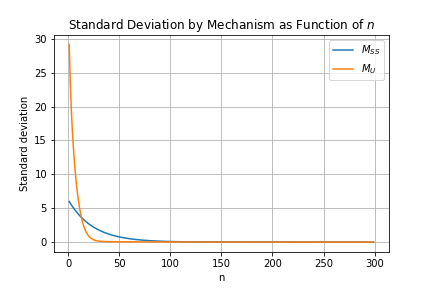}
    \caption{Standard deviations of the mechanisms $M_{SS}$ and $M_U$ for a mean of $n$ records in [0,1]. The mean is fixed at .5 and the mechanisms have a privacy budget of $\epsilon_2 = .5$.}
    \label{fig:MnAppSDsBig}
\end{figure}

\begin{figure}[H]
    \centering
    \includegraphics[width=0.5\linewidth]{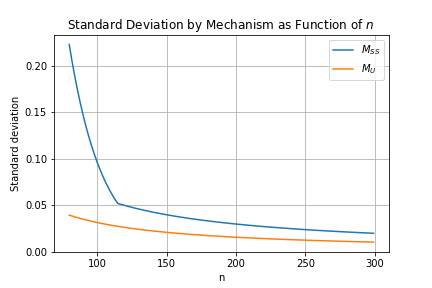}
    \caption{Same as Figure~\ref{fig:MnAppSDsBig}, zoomed in to larger values of $n$ (see the horizontal axis endpoints). Standard deviations of the mechanisms $M_{SS}$ and $M_U$ for a mean of $n$ records in [0,1]. The mean is fixed at .5 and the mechanisms have a privacy budget of $\epsilon_2 = .5$.}
    \label{fig:MnAppSDsZoom}
\end{figure}

\begin{figure}[H]
    \centering
    \includegraphics[width=0.5\linewidth]{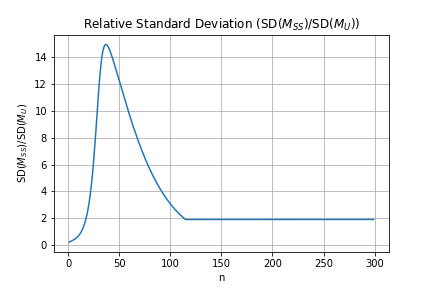}
    \caption{Standard deviation of $M_{SS}$ divided by the standard deviation of $M_U$ for a mean of $n$ records in [0,1]. The mean is fixed at .5 and the mechanisms have a privacy budget of $\epsilon_2 = .5$.}
    \label{fig:MnAppSDsRel}
\end{figure}


\section{Application to Slowly Scaling PRDP}
\label{sec:PRDP}

In this section, we use our estimators to develop versions of unbiased privacy mechanisms from \cite{finleyEtAl2024} that enjoy stronger privacy guarantees. Our estimators allow us to do so while maintaining the mechanisms' unbiasedness.

$(\epsilon,0)$-DP guarantees the upper bound $\epsilon$ on the privacy loss between any pair of neighboring databases. \cite{finleyEtAl2024} develops mechanisms for a related privacy concept, per-record DP (PRDP) \cite{seemanEtAl2024}. PRDP generalizes ($\epsilon,0$)-DP by letting the privacy loss bound be a function of the record on which a given pair of neighbors differs. Semantically, this allows different records to have different levels of protection.

Denote a record in the database $D$ by $r$ and denote $r$'s attribute $c$ by $r.c \in [0,\infty)$. PRDP was originally motivated by the need to protect data used to compute the sum query $q(D) = \sum_{r.c \in D} r.c$. Because the domain of $c$ is unbounded, this sum can change by an arbitrarily large amount when a record is added or deleted. That is, the sum's global sensitivity is infinite. This prevents commonly used privacy mechanisms, such as the Laplace mechanism, from providing a differential privacy guarantee with finite $\epsilon$.

The traditional fix for this is to clip attribute $c$ to lie in a bounded set before taking the sum. $(\epsilon,0)$-DP can then be guaranteed by perturbing the sum with noise scaled in proportion to the width of the clipped data's domain. Unfortunately, when the sum is dominated by a small number of large outliers, the outliers typically need to be clipped to drastically smaller values to preserve a reasonable balance of privacy loss and noise variance. This can induce catastrophic bias, rendering the clipped sums essentially useless. One might expect to see this type of behavior with income data, for example. 

PRDP allows us to take a finer look at the privacy-utility tradeoff by recognizing that, even though outliers may suffer extreme privacy loss, the rest of the dataset may still enjoy strong privacy protections. Intuitively, a particular record's privacy loss is proportional only to the amount by which the addition or deletion of \textit{that record} can change the query. Queries may be highly sensitive to the presence of outliers while being relatively insensitive to typical records, leading different records to have different levels of privacy loss. The reassurance that the vast majority of the data may enjoy strong privacy guarantees whether or not the data is clipped may allow a data curator to reasonably decide against clipping if the resulting bias outweighs the enhanced privacy protection for a small number of records.

Below, we define PRDP.

\begin{definition}[$P$-Per-Record Differential Privacy ($P$-PRDP) \cite{seemanEtAl2024, finleyEtAl2024}]
\label{def:PRDP}
Let $\ominus$ denote the symmetric set difference. The mechanism $M$ satisfies per-record differential privacy with the policy function $P$ ($P$-PRDP) if, for any record $r$; any pair of neighboring databases $D, D'$ such that $D \ominus D' = \{r\}$; and any measurable set of possible outcomes $S$, we have
$$
\Pr[\mathcal{M}(D) \in S] < e^{P(r)} \cdot \Pr[\mathcal{M}(D') \in S].
$$
\end{definition}

Ensuring strong privacy protection corresponds to ensuring that $P$ is, in some sense, small. $(\epsilon,0)$-DP is recovered by making the constant privacy guarantee $P(r) = \epsilon$ for all $r$, and strong privacy protection follows from a small $\epsilon$. We cannot always make a guarantee this strong. Take the example where we want to publish a sum query on an unbounded attribute $c$ that we are unwilling to clip. In this case, the privacy loss of the mechanisms that we will consider here is growing in $r.c$. Even though we cannot prevent $P(r)$ from growing without bound in $r.c$, we can use mechanisms for which the growth rate is slow. \cite{finleyEtAl2024} call such mechanisms ``slowly scaling.'' A slowly growing $P$ narrows the gap in privacy losses between records with large and small values of $c$, letting a data curator more easily provide a desired level of protection for the bulk of the data without compromising too much on the privacy of outliers.

\cite{finleyEtAl2024} introduces slowly scaling mechanisms, called transformation mechanisms, that work by adding Gaussian noise to a concave transformation $f$ of the query (plus offset term) $q(D)+a$ and then feeding the noisy value of $f(q(D)+a)$ to an estimator $g$ of $q(D)$. By adding Gaussian noise, these mechanisms satisfy per-record zero-concentrated DP (PRzCDP), which is a weaker privacy guarantee than PRDP. PRzCDP relates to zero-concentrated DP \cite{BunStienke2016zCDP} in the same way that PRDP relates to $\epsilon$-DP. The use of Gaussian noise also allowed \cite{finleyEtAl2024} to draw on existing unbiased estimators from \cite{WashioEtAl1956} to make their mechanism unbiased for a variety of transformation functions $f$.

Using the unbiased estimators from Theorem~\ref{thm:unbiased}, we strengthen the transformation mechanisms to provide PRDP guarantees by adding Laplace, rather than Gaussian, noise, and we do so without losing the mechanism's unbiasedness. Algorithm~\ref{alg:transform} lays out our transformation mechanism.

\begin{algorithm}[H]
\caption{\label{alg:transform} PRDP Transformation Mechanism}
\begin{algorithmic}[1]
    \Procedure{TransformationPrivatizeLap}{Private query answer $q(D)$,  offset parameter $a$, \newline scale parameter $b$, transformation function $f:[a,\infty)\to \mathcal{F} \subseteq \R$, estimator $g:\mathcal{F} \to \mathcal{G} \subseteq \R$} 
    \State $v \leftarrow f(q(D) + a)$
    \State $\tilde{v} \leftarrow \Lap(v,b)$ \label{line:transform_noise}
    \State $\tilde{S}\leftarrow g(\tilde{v})$
    \State \textbf{output} $\tilde{S}$
    \EndProcedure
\end{algorithmic}    
\end{algorithm}

To obtain the PRDP guarantee of Algorithm~\ref{alg:transform}, we first need to define the per-record sensitivity \cite{finleyEtAl2024}, a record-specific analog of the global sensitivity.

\begin{definition}[Per-Record Sensitivity \cite{finleyEtAl2024}] \label{def:PR-sensitivity}
The per-record sensitivity of the univariate, real-valued query $q$ for record $r$ is $$\Delta(r) \equiv \sup_{D,D^\prime \text{ such that } D \ominus D^\prime = \{r\}} |q(D) - q(D^\prime)|.$$
\end{definition}

Theorem~\ref{thm:transformPRDP} gives our most generic result on the PRDP guarantees of Algorithm~\ref{alg:transform}. Theorems~\ref{thm:transformPRDP} and \ref{thm:privTrnsSum}, their proofs, as well as Algorithm~\ref{alg:transform} are minimally modified from their analogs in \cite{finleyEtAl2024}, which use Gaussian, rather than Laplace noise. This is to facilitate the interested reader's comparison of our results with theirs.

\begin{theorem}[PRDP Guarantee for Transformation Mechanisms]
\label{thm:transformPRDP}
    Assume the query value $q(D) \in [0,\infty)$; the offset parameter $a \in \R$; the noise scale parameter $b \in (0,\infty)$; the transformation function $f:[a,\infty) \to \mathcal{F} \subseteq \R$ is concave and strictly increasing; and the estimator $g:\mathcal{F} \to \mathcal{G} \subseteq \R$. Denote by $\Delta_f(r)$ the per-record sensitivity of the query $f(q(D) + a)$, as defined in Definition~\ref{def:PR-sensitivity}. Algorithm~$\ref{alg:transform}\allowbreak(q(D), a, b, f, g)$ satisfies $P$-PRDP for $P(r) = \frac{\Delta_f(r)}{b}$.
\end{theorem}

See Appendix~\ref{app:PRDPProofs} for the proof.

Probably the most common query encountered in applications of formal privacy, even as a component of other, larger queries, is the sum query. We now use the above result to derive the policy function for the transformation mechanism applied to a sum query.

\begin{theorem}[Privacy of Transformation Mechanisms for Sum Query]
\label{thm:privTrnsSum}
Let the assumptions of Theorem~\ref{thm:transformPRDP} hold, and further assume $a \geq 0$ and $r.c \geq 0$ for all records $r$. For the sum query $q(D) = \sum_{r \in D} r.c$, the per-record sensitivity of $f(q(D) + a)$ is $\Delta_f(r) = f(r.c + a) - f(a)$ and Algorithm~\ref{alg:transform}($q(D),a,b,f,g$) satisfies PRDP with the policy function $P(r) = [f(r.c + a) - f(a)]/b$.
\end{theorem}

See Appendix~\ref{app:PRDPProofs} for the proof.

Critically, the policy function from Theorem~\ref{thm:privTrnsSum} grows in $r.c$ more slowly when the transformation function $f$ grows more slowly. In the case where $a = 0$ and $f(x) = \sqrt[k]{x}$ for some $k \geq 1$, the policy function is simply $\sqrt[k]{r.c}$. Choosing larger values of $k$, then, forces the privacy loss to grow more slowly in $r.c$, reducing the gap in privacy losses between records with large and small values of $c$.

Applying our main results, we can obtain estimators such that the transformation mechanism gives us an unbiased estimate of $q(D)$. In particular, polynomials\footnote{In the case of polynomials, \cite{hillebrand2023unbiased} derived unbiased estimators which could also be used here. The estimator obtained from our Theorem~\ref{thm:unbiased} merely simplifies computation in this setting.} are twice differentiable functions which are tempered distributions, so the following holds for $f(x) = \sqrt[k]{x}$:
\begin{corollary}
    Given any function $f$ such that $f^{-1}$ satisfies the conditions in Theorem~\ref{thm:unbiased}, $a \in \RR$, $b \in (0,\infty)$, estimator $g:\mathcal{F} \rightarrow \mathcal{G} \subseteq \RR$, and $r.c \ge 0$ for all records $r$, there exists an unbiased estimator for $q(D)$ satisfying $P$-PRDP for $P(r) = \left[f(r.c + a)-f(a)\right]/b$.
\end{corollary}

\begin{proof}
    The conditions for Theorem~\ref{thm:privTrnsSum} hold, and Theorem~\ref{thm:unbiased} gives us a function $g$ such that $\E[g(\tilde{v})] = q(D) + a$ is an unbiased estimator for $f^{-1}$. Therefore, $g(\tilde{v}) - a$ is an unbiased estimator for $q(D)$.
\end{proof}

\section{Polynomial Functions under General Noise Distributions}
\label{sec:unbiasPoly}

Additive mechanisms other than the Laplace mechanism, such as the discrete Gaussian or the staircase mechanisms, may be preferable in practice due to achieving higher accuracy while having similar privacy loss \cite{CanonneEtAl2020DGM, staircase}.
In contrast to the Laplace case, these mechanisms may not admit tractable Fourier transforms, and hence unbiased estimators are generally not available in closed form.
One exception is when the query of interest is a polynomial in one or many queries. 

Using the following results to obtain an unbiased estimator of a polynomial that approximates a non-polynomial estimand may also allow users to obtain approximately unbiased estimators with great generality.

\begin{theorem}
\label{thm:unbiased-uni-poly}
    Suppose a mechanism takes as input $ q \in \RR$ and outputs $\tilde q = q+Z$ for a random variable $Z$ with at least $p$ finite, publicly known moments.
    If $f(q)$ is a polynomial in $q$ of degree at most $p$, there exists an unbiased estimator $g(\tilde q)$ of $f(q)$, which is itself a polynomial of degree at most $p$ and is available in closed form.
\end{theorem}

\begin{proof}
Suppose $f(q) = \sum_{n=0}^p b_n q^n$.
Let us find an unbiased estimator of the form $g(\tilde q) = \sum_{n=0}^p a_n \tilde q^p$. We have
\begin{align}
g(q+z) = \sum_{n=0}^p a_n \sum_{k=0}^n \binom{n}{k} q^k z^{n-k}.
\end{align}
Denote $\mu_r = \E[z^r]$ and take expectations to obtain
\begin{align}
\E[g(q+z)] &= \sum_{n=0}^p a_n \sum_{k=0}^n \binom{n}{k} q^k \mu_{n-k} \\    
&= \sum_{n=0}^p \binom{n}{0} a_n\mu_n + \left(\sum_{n=1}^p \binom{n}{1} a_n \mu_{n-1} \right) q + \left(\sum_{n=2}^p \binom{n}{2} a_n \mu_{n-2} \right) q^2 + \cdots 
\end{align}
For this polynomial to be equal to $f$, we need $a=(a_0,\dots,a_p)'$ to solve $M a = b$, where $b=(b_0,\dots,b_p)'$ and
\begin{align}
    M = 
    \begin{pmatrix}
        \mu_0 & \mu_1 & \mu_2 & \mu_3 & \cdots & \mu_p \\
        0     & \binom{1}{1} \mu_0 & \binom{2}{1} \mu_1 & \binom{3}{1} \mu_2 & \cdots & \binom{p}{1} \mu_{p-1} \\
        0 & 0 & \binom{2}{2}\mu_0 & \binom{3}{2} \mu_1 & \cdots & \binom{p}{2}\mu_{p-2} \\
        0 & 0 & 0 & \binom{3}{3} \mu_0 & \cdots & \binom{p}{3} \mu_{p-3} \\
        \vdots & \vdots & \vdots & & \ddots & \vdots \\
        0 & 0 & 0 & 0 & \cdots & \binom{p}{p} \mu_0
    \end{pmatrix}
    .
\end{align}
Clearly, $M$ is nondegenerate, and so the desired coefficients $a$ exist and are unique.
\end{proof}

We now extend this result to polynomials in multiple (univariate) queries, assuming that the noise variables added to each query are independent. The latter assumption, while seemingly restrictive, is typical for additive noise mechanisms in differential privacy.
\begin{theorem}
    Suppose a mechanism takes as input $(q_1,\dots, q_m)$ and outputs $(\tilde q_1,\dots,\tilde q_m) = (q_1+Z_1, q_2+Z_2, \dots, q_m+Z_m)$ for independent random variables $Z_1,\dots,Z_m$ with finite, publicly known moments.
    If $f(q_1,\dots,q_m)$ is a polynomial in $(q_1,\dots,q_m)$, there exists an unbiased estimator $g(\tilde q_1,\dots,\tilde q_m)$ of $f(q)$, which is itself a polynomial available in closed form.
\end{theorem}
\begin{proof}
Clearly, it suffices to derive unbiased estimators for $f(q_1,\dots, 
q_m)=\prod_{i=1}^m q_i^{p_i}$.
Let $g_i(\tilde q_i)$ be the unbiased estimator of $q_i^{p_i}$ as in Theorem \ref{thm:unbiased-uni-poly} and set $g(\tilde q_1,\dots,\tilde q_m) = \prod_{i=1}^m g_i(\tilde q_i)$.
Since $g_1(\tilde q_1),\dots,g_m(\tilde q_m)$ are independent random variables, we have
\begin{align}
    \E[g(\tilde q_1,\dots \tilde q_m)]= \E\left[\prod_{i=1}^m g_i(\tilde q_i)\right] = \prod_{i=1}^m \E\left[g_i(\tilde q_i)\right] = \prod_{i=1}^m q_i^{p_i} = f(q_1,\dots,q_m).
\end{align}
\end{proof}

\section{Conclusions and Future Work}
In this work, we have shown how to compute unbiased estimators of twice-differentiable tempered distributions evaluated on privatized statistics with added Laplace noise. In addition, we have proposed a method to extend this result to twice-differentiable functions which are not tempered distributions in a way that achieves approximately optimal expected squared error.

As the Laplace distribution is a commonly used and simple DP mechanism, these results are widely applicable to obtain unbiased statistics for free in postprocessing, which is particularly valuable due to the fact that aggregating unbiased statistics accumulates error more slowly than aggregating biased statistics. We have applied our results to derive a competitive unbiased algorithm for means and to derive unbiased transformation mechanisms for per-record DP mechanisms that enjoy stronger privacy protection than do analogs in previous work. Finally, we have derived an unbiased estimator for polynomials under arbitrary noise distributions with known moments, such as the discrete Gaussian mechanism or the staircase mechanism \cite{CanonneEtAl2020DGM, staircase}.

We believe this paper opens several avenues for future research. These include the use of the deconvolution method to obtain unbiased estimators for other estimands and noise distributions. We believe a deconvolution method using multivariate Fourier transforms could also be used to obtain unbiased estimators of functions of multivariate queries. Although we did not attempt to optimize the numerical implementation in Section~\ref{sec:Numerical_Means} of the integration in Section~\ref{sec:unbiasLaplaceNonTempDist}, we believe that an improved implementation could enable the practical use of higher-order polynomial extensions and further reduce error. In Section~\ref{sec:unbiasPoly}, we developed estimators that are exactly unbiased for polynomials that could approximate other functions of interest. Further work could elaborate on this process, developing concrete procedures for picking the approximating polynomial and deriving bounds on the resulting bias. Finally, future work could attempt to derive noise distributions that are optimal in the sense of minimizing the variances (or other utility metrics) of their unbiased estimators.

\appendix

\section{Distribution Theory Background and Proofs of Previously Known Results \label{app:Prelim}}

\subsection{Proof of the Convolution Theorem}
Recall the convolution theorem:
\begin{theorem}
    (\cite{vanDijk+2013} section 7.1 property c)
    $$\widehat{(f*g)}(y) = \hat{f}(y) \cdot \hat{g}(y)$$
\end{theorem}

\begin{proof}
    We start by writing out the definition of the LHS:
    \begin{align}
        \widehat{(f*g)}(y) &= \int_{-\infty}^\infty e^{-2\pi i x y} \int_{-\infty}^\infty f(z) g(x-z) dz dx\\
        &= \int_{-\infty}^\infty dz \int_{-\infty}^\infty e^{-2\pi i xy}f(z) g(x-z) dx. \quad(\text{Fubini's theorem})
    \end{align}
    Here, Fubini's theorem \cite{Fubini} allows us to flip the order of integration if the integral is finite, the measures are $\sigma$-finite, and the integrand is measurable.

    And then for the RHS:
    \begin{align}
        \hat{f}(y) \cdot \hat{g}(y) &= \int_{-\infty}^\infty e^{-2\pi i xy} f(x) dx \cdot \int_{-\infty}^\infty e^{-2\pi i zy} g(z) dz\\
        &= \int_{-\infty}^\infty \int_{-\infty}^\infty e^{-2\pi i xy} f(x)e^{-2\pi i zy} g(z) dx dz\\
        &= \int_{-\infty}^\infty \int_{-\infty}^\infty e^{-2\pi i (x+z)y} f(x)g(z) dx dz\\
        &=\int_{-\infty}^\infty \int_{-\infty}^\infty e^{-2\pi i uy} f(x) g(u-x) dx dz.
    \end{align}
    We get the iterated integral in the 2nd line by noting that the two integrals are constant with respect to each other. Then, note that if we simply change the variables $x \rightarrow z$ and $u \rightarrow x$, we get back exactly the same expression as the LHS. Thus,
\begin{align}
    \widehat{(f*g)}(y) = \hat{f}(y) \cdot \hat{g}(y).
\end{align}
\end{proof}

\subsection{Proof of the Derivative Property}

Recall the Derivative Property:
\begin{theorem} \label{thm:funcFourierDerivAppendix} (\cite{vanDijk+2013} Theorem 7.6)
For absolutely integrable functions $f$ that are $k$ times continuously differentiable and whose derivatives of order ${1,...,k}$ are absolutely integrable, $F\left[\frac{d^k}{dx^k} f(x)\right](y) = (2 \pi i y)^k \hat{f}(x)$.
\end{theorem}

\begin{proof}
We simply use integration by parts:
    \begin{align*}
        F[f'(x)](y) &= \int_{-\infty}^\infty e^{-2\pi i xy} f'(x) dx\\
        &= e^{-2\pi i xy} f(x)|_{-\infty}^\infty + \int_{-\infty}^\infty 2 \pi i y e^{-2\pi i xy} f(x) dx\\
        &= 0 + 2\pi i y \int_{-\infty}^\infty e^{-2\pi i xy} f(x) dx\\
        &= 2 \pi i y \hat{f}(y).
    \end{align*}

    Repeating this procedure $k$ times gives the theorem.
\end{proof}

\subsection{Distribution Theory}
\label{app:DistTheory}

In order to generalize the Fourier transform to functions that are not absolutely integrable, we must introduce some concepts from distribution theory \cite{vanDijk+2013}:

\begin{definition} \label{def:testFunc}
    \item A test function is a function $\varphi:\RR\rightarrow \C$ which is in $C^\infty$ (that is, it is infinitely differentiable) with bounded support.
\end{definition}

\begin{example}
    One example of a test function is $f(x) = 1$ if $x \in [0,1]$ and $0$ elsewhere. 
\end{example}

Now, we can define distributions, which are a generalization of functions in the sense that we can construct a distribution corresponding to any function that acts like a function, as described below. There are also distributions that are not functions, an example of which will be described after we provide the definition:

\begin{definition} \label{def:Distribution}
    \item A distribution is a linear functional (that is, a linear mapping from functions to $\RR$) that acts on test functions as follows.

    For a distribution $T$, we have:

    \begin{enumerate}
        \item $T(\varphi_1) + T(\varphi_2) = T(\varphi_1 + \varphi_2)$ for all test functions $\varphi_1, \varphi_2$.
        \item $\lambda T(\varphi) = T(\lambda \varphi)$ for all test functions $\varphi$ and scalars $\lambda$.
        \item If a sequence $\lim_{n \rightarrow \infty} \varphi_n = \varphi$, then
        $$\lim_{n \rightarrow \infty} T(\varphi_i) = T(\varphi).$$
    \end{enumerate}

\end{definition}

In distribution theory, we typically represent $T(\varphi)$ as $\langle T, \varphi\rangle$, and for any locally integrable function $f$ (that is, its integral converges on any bounded domain), there exists a distribution $T_f$ where

$$\langle T_f, \varphi\rangle = \int_{-\infty}^\infty f(x) \varphi(x) dx.$$

Distributions generalize functions in the sense that they capture the behavior of functions in the above integral but include other linear functionals that cannot be expressed via such an integral.

\begin{example}
\label{ex:dirac}
    The Dirac delta distribution $\delta$ is a distribution where $\langle \delta, f\rangle = f(0)$ for all functions $f$, but there is no function that satisfies this property.
\end{example}

Recall that we required to define the Schwartz space and tempered distributions in order to extend Fourier transforms to functions that are not absolutely integrable. Tempered distributions are a special case of distributions that have to act on all Schwartz functions, which are a superset of test functions.

To see this, recall the definition of the Schwartz Space:

\begin{definition}
    \item The Schwartz space $S(\RR)$ is defined as follows:
    $$S(\RR) = \{s:\RR \rightarrow \C \mid s \in C^\infty, \sup_{x \in \RR} |x^m s^{(n)}(x)| < \infty \quad \forall m,n \ge 0\}.$$
    That is, functions in $S(\RR)$ are infinitely differentiable everywhere and all of their derivatives go to 0 at a super-polynomial rate.
\end{definition}

A function with tails that are equal to 0 satisfies this property as, if $s^{(n)}(x) = 0$ for $x < L$ or $x > U$ for some lower bound $L$ and upper bound $U$, and $s$ is well-defined, then both $s^{(n)}(x)$ and $x^m$ must be finite for any $x$ such that $x^ms^{(n)}(x) \ne 0$. Thus, we have $|x^m s^{(n)}(x)| < \infty \: \forall m,n \ge 0$.

Then, we have the following definition of tempered distributions:

\begin{definition} \label{def:tmprDistAppendix}
    The space of tempered distributions $S'(\RR)$ is the set of distributions $T$ such that $\langle T, s\rangle \in \C$ for all $s \in S(\RR)$. Additionally, for any function $f$ we define
    $$\langle T_f, s\rangle = \int_{-\infty}^\infty f(x) s(x)dx.$$
\end{definition}

Note that $T_f$ is a distribution corresponding to the function $f$, but is itself not a function. This means it acts on functions but not numbers. However, we need to be able to use it as a function when we return it as our unbiased estimator so we simply define $T_f(q) = f(q) \quad \forall q \in \RR$. It is similar in this sense that we say a function $f$ is a tempered distribution if $T_f \in S'(\RR)$, even though $f$ is technically not a distribution as it acts on the wrong objects.

Like before, we can define the Fourier transform of tempered distributions as follows:

\begin{definition} \label{def:DistFourierAppendix}
    $\hat{T}$ is the distribution such that for all Schwartz functions $\varphi$,
    $$\langle \hat{T}, \varphi\rangle = \langle T, \hat{\varphi}\rangle.$$
\end{definition}

Note that, since distributions are not necessarily functions, their derivatives cannot be defined in the traditional sense. Instead, we consider the \textit{distributional derivative} $T'$, defined as the distribution such that $\langle T', \varphi\rangle = -\langle T, \varphi'\rangle$. This is similar to the usual notion of a function's derivative in that it preserves integration by parts:
\begin{align} \label{eq:distDeriv}
    \langle T', \varphi\rangle = \int_{-\infty}^\infty T'(x) \varphi(x) dx = T(x)\varphi(x)|_{-\infty}^\infty - \int_{-\infty}^\infty T(x) \varphi'(x) = -\langle T, \varphi'\rangle,
\end{align}
where $T(x)\varphi(x)|_{-\infty}^\infty = 0$ because $\varphi(x) = 0$ at $\pm \infty$.

The following analogue of Theorem~\ref{thm:funcFourierDerivAppendix} extends that derivative property of the Fourier transform to tempered distributions.

\begin{corollary}
\label{cor:derivative} (\cite{vanDijk+2013} section 7.8 Example 5)
For all $T \in S'(\RR)$,
    $$\widehat{T^{(n)}} = (2\pi i y)^n \hat{T}.$$
\end{corollary}

\begin{proof}
Note that for $f \in S'(\RR)$ and $\varphi \in S(\RR)$, we have
\begin{align}
    \langle \hat{f'}, \varphi\rangle = \langle f', \hat{\varphi}\rangle = -\langle f, \left(\hat{\varphi}\right)'\rangle.
\end{align}

By the dual of the derivative property on Schwartz functions, we have $\widehat{x \varphi(x)} = \frac{i}{2\pi}\left(\hat{\varphi}\right)'$, or
\begin{align}
    \widehat{-2\pi i x \varphi(x)} = \left(\hat{\varphi}\right)'.
\end{align}

By linearity of the inner product and definition of the Fourier transform on tempered distributions, we thus have
\begin{align}
    -\langle f, \left(\hat{\varphi}\right)'\rangle = \langle f, \widehat{2\pi ix \varphi(x)}\rangle = \langle \hat{f}, 2\pi i x \varphi(x)\rangle = \langle 2\pi i x \hat{f}(x), \varphi \rangle.
\end{align}

\end{proof}

With these, we can generalize Theorem~\ref{thm:unbiased} to all tempered distributions instead of just the functions that correspond to them.


\section{Proof for Section 3} \label{App: section3}
\begin{theorem}
(Restatement of Theorem~\ref{thm:unbiased})
    \item Let $\tilde{q} \sim q + \Lap(0,b)$. 
    \begin{itemize}
         \item For any twice-differentiable function $f:\R \to \R$ that is a tempered distribution, $f(\tilde{q}) - b^2 f''(\tilde{q})$ is an unbiased estimator for $f(q)$.
        \item For any function $f:\R \to \R$, if two estimators $g_1(\tilde{q})$ and $g_2(\tilde{q})$ are unbiased for $f(q)$, then $g_1$ and $g_2$ are equal almost everywhere.
    \end{itemize}
\end{theorem}

\begin{proof}
    We start by proving the first item in Theorem~\ref{thm:unbiased}. From the table of Fourier transforms in \cite{Kammler_2008}, we know that
\begin{align}
    F\left[e^{-a|x|}\right](y) = \frac{2a}{4\pi^2 y^2 + a^2}
\end{align}

and the PDF of $\Lap(0,b)$ is $\frac{1}{2b}e^{-|x|/b}$, so
\begin{align}
    F\left[\frac{1}{2b} e^{-|x|/b}\right](y) = \frac{1}{2b} \cdot \frac{2b}{4\pi^2 y^2 b^2 + 1} = \frac{1}{1+ 4\pi^2 y^2 b^2}.
\end{align}

Recall that the unbiased estimator for a function $f$ would then be the function that satisfies
\begin{align}
    f(q) = \E[g(\tilde{q})] = \int_{-\infty}^\infty g(\tilde{q}) P(\tilde{q} - q) d\tilde{q} = (f * P)(q),
\end{align}

where $P$ is the PDF of $\Lap(0,b)$.

Thus, by Theorem~\ref{thm:convolution}, we can take the Fourier transform of both sides and get
\begin{align}\hat{f}(y) &= \hat{g}(y) \cdot \hat{P}(y)\\
\hat{g}(y) &= \frac{\hat{f}(y)}{\hat{P}(y)}\\
g(\tilde{q}) &= F^{-1} \left[\frac{\hat{f}(y)}{\hat{P}(y)}\right](\tilde{q})\\
&= F^{-1} \left[\hat{f}(y)(1+4\pi^2y^2b^2)\right](\tilde{q}).
\end{align}
Now, note that applying Theorem~\ref{thm:tmpDistFourierDeriv} gives us $4\pi^2y^2b^2\hat{f} = -b^2\widehat{f''}$. Combining this with the linearity of the Fourier transform, the inverse Fourier transform gives us:
\begin{align}
    g(\tilde{q}) = f(\tilde{q}) - b^2 f''(\tilde{q}).
\end{align}

We now prove the second item in Theorem~\ref{thm:unbiased}. Here, ``distribution'' will refer to probability distributions, rather than generalized functions. Page 188 of \cite{OosterhoffEtAl1987Complete} shows that the location family of Laplace distributions is \textit{complete}, meaning that a function $h$ must be equal to 0 almost everywhere in order to satisfy 
\begin{align}
    \E[h(\tilde{q})|q] = \int_{-\infty}^{\infty} h(x) \frac{1}{2b}e^{-\frac{|x-q|}{b}} dx = 0
\end{align}
for all $q \in \R$. Intuitively, this means that the only function unbiased for the zero function is the zero function.

By a variation on the Lehmann-Scheff\'e theorem \cite{LehmannScheffe1950}, this implies our uniqueness result. The proof, based on \cite{Mackey2015LectNotes}, is straightforward: Suppose that the two estimators $g_1$ and $g_2$ satisfy $\E[g_1(\tilde{q})|q] = \E[g_2(\tilde{q})|q] = f(q)$ for all $q \in \R$. It follows that $
    \E[g_1(\tilde{q}) - g_2(\tilde{q})|q] = 0.$
Completeness then implies that $g_1(x) - g_2(x) = 0$ for $x$ almost everywhere, giving us the result.

For the reader trying to prove similar uniqueness results when other noise distributions are used, we note that \cite{Mattner1992Complete} provides sufficient conditions for completeness of location families of other distributions. Completeness is not a given, though; \cite{Mattner1993Incomplete} shows that, except for the Gaussian distribution, distributions with tails thinner than a Laplace distribution typically have incomplete location families.
\end{proof}

\section{Proofs for Section 4} \label{App:section4}

The following lemma is required for the proof of Theorem~\ref{thm:polyApprox}.

\begin{lemma}[Constrained Polynomial Approximation] \label{lem:finalApproxLemma}
        Let $L \in \R$ and $c > 0$. Let $f:[L,\infty) \to \R$ be twice differentiable and denote the weight function $W_{L,c}(x) \equiv
        e^{-c|x-L|}$. Let $h:(-\infty,L] \to \R$ be a twice continuously differentiable tempered distribution satisfying $h(L) = f(L)$, $h'(L) = f'(L)$ and $h''(L) = f''(L)$. There exists a sequence of polynomials $(p_K)_{K=1}^\infty$ that satisfy $p_K(L) = f(L)$, $p_K'(L) = f'(L)$ and $p_K''(L) = f''(L)$ such that the following limits hold for $p \in [1,\infty) \cup \{\infty\}$:

        \begin{align}
        \lim_{K \to \infty} ||(p_K - h) W_{L,c}||_{L_p((-\infty,L))} = \lim_{K \to \infty}||(p''_K(x) -h'')W_{L,c}||_{L_p((-\infty,L))} = 0
        \end{align}
    \end{lemma}

\begin{proof}
Our proof relies critically on results on weighted approximation by polynomials over the real line. In particular, we use the following result (see, for example, Corollary 1.5 and Theorem 1.6 of \cite{Lubinsky2007Survey}):

\begin{theorem}[Weighted polynomial approximation]
\label{thm:wgtPolyApprox}
    Let $p \in [1,\infty) \cup \{\infty\}$ and let the weight function $W(x) = e^{-|x|}$. For any continuous, measurable function $h:\R \to \R$ satisfying $||h W||_{L_p(\R)} < \infty$, there exists a sequence of polynomials $(p_K)_{K=1}^\infty$ such that
    \begin{align}
       \lim_{K \to \infty} ||(h - p_K)W||_{L_p(\R)} = 0.
    \end{align}
\end{theorem}

We first note that Theorem~\ref{thm:wgtPolyApprox} can be immediately generalized to use the weight functions $W_{L,c} \equiv e^{-c|x-L|}$ and the $||\cdot||_{L_p((-\infty,L])}$ norm. The weight function can be changed by simply using Theorem~\ref{thm:wgtPolyApprox} to obtain a polynomial sequence $(\hat{p}_K)_{K=1}^\infty$ that approximates the translated and rescaled function $\hat{h}(x) \equiv h(c^{-1}x + L)$. The polynomial sequence $(p_K)_{K=1}^\infty$ defined such that $p_K(x) = \hat{p}_K(c^{-1}x + L)$ then approximates $h$. More formally, for $p \in [1,\infty)$, we have the following:

\begin{align}
    ||(h - p_K)W_{L,c}||_{L_p((-\infty,L])} 
    &= \left( \int_{-\infty}^L |(h(x) - p_K(x)) e^{-c|x-L|}|^p dx \right)^{1/p} \\
    &= \frac{1}{c^{1/p}} \left( \int_{-\infty}^0 |(h(\frac{u}{c} + L) - p_K(\frac{u}{c} + L)) e^{-|u|}|^p du \right)^{1/p} \\
    &\leq \frac{1}{c^{1/p}} \left( \int_{-\infty}^\infty |(h(\frac{u}{c} + L) - p_K(\frac{u}{c} + L)) e^{-|u|}|^p du \right)^{1/p} \\
    &= \frac{1}{c^{1/p}} ||(\hat{h} - \hat{p}_K)W||_{L_p(\R)}.
\end{align}

For the $p = \infty$ case, note that $(h-p_K)W_{L,c}$ is continuous, and so its essential supremum is simply its supremum. It follows that

\begin{align}
    ||(h - p_K)W_{L,c}||_{L_\infty((-\infty,L])} 
    &= \sup_{x \in (-\infty,L]} |(h(x) - p_K(x))W_{L,c}(x)| \\
    &= \sup_{u \in (-\infty,0]} |(h(\frac{u}{c} + L) - p_K(\frac{u}{c} + L))W_{L,c}(\frac{u}{c} + L)| \\
    &\leq \sup_{u \in (-\infty,\infty)]} |(h(\frac{u}{c} + L) - p_K(\frac{u}{c} + L))W_{L,c}(\frac{u}{c} + L)| \\
    &= ||(\hat{h} - \hat{p}_K)W||_{L_\infty(\R)}.
\end{align}

Since $(\hat{p}_K)_{K=1}^\infty$ is chosen to satisfy Theorem~\ref{thm:wgtPolyApprox} for $\hat{h}$, we have $\lim_{K \to \infty} \frac{1}{c} ||(\hat{h} - \hat{p}_K)W||_{L_p(\R)} = 0$. This gives us the following approximation result: 

\begin{align}\label{eq:wgtPolyApproxGeneral}
    \lim_{K \to \infty} ||(h - p_K)W_{L,c}||_{L_p((-\infty,L])} = 0.
\end{align}

Let $(\tilde{p}''_K)_{K=1}^\infty$ be a sequence of polynomials that approximate $h''$ in the sense of Equation~\ref{eq:wgtPolyApproxGeneral}, with the weight function $W_{L,4c}$. Because $(h'' - \tilde{p}''_K)W_{L,4c}$ is a continuous function, Equation~\ref{eq:wgtPolyApproxGeneral} with $p = \infty$ implies the following pointwise limit for all $x \in (-\infty, L]$.

\begin{align}
    \lim_{K \to \infty} (h''(x) - \tilde{p}''_K(x))W_{L,4c}(x) = 0
\end{align}

We define the polynomial sequence $(p''_K)_{K=1}^\infty$ by $p''_K \equiv \tilde{p}''_K - \tilde{p}''_K(L) + h''(L)$. Note that $p''_K(L) = h''(L)$ and $p''_K$ still approximates $h''$ in the sense of Equation~\ref{eq:wgtPolyApproxGeneral} because
\begin{align} 
    ||(h'' - p''_K)W_{L,4c}||_{L_p((-\infty,L])}
    &= ||(h'' - p''_K)W_{L,4c}||_{L_p((-\infty,L])} \\
    &= ||(h'' - \tilde{p}''_K + \tilde{p}''_K(L) - h''(L))W_{L,4c}||_{L_p((-\infty,L])} \\
    &\leq ||(h'' - \tilde{p}''_K)W_{L,4c}||_{L_p((-\infty,L])} + (\tilde{p}''_K(L) - h''(L))||W_{L,4c}||_{L_p((-\infty,L])}.
\end{align}

Clearly, 
\begin{align}
    \lim_{K \to \infty} ||(h'' - \tilde{p}''_K)W_{L,4c}||_{L_p((-\infty,L])} + (\tilde{p}''_K(L) - h''(L))||W_{L,4c}||_{L_p((-\infty,L])} = 0,
\end{align}
implying, in turn, our approximation result: 
\begin{align} \label{eq:approxRsltDeriv2}
    \lim_{K \to \infty} ||(h'' - p''_K)W_{L,4c}||_{L_p((-\infty,L])} = 0.
\end{align}

We now integrate $p''_K$ twice to obtain a polynomial sequence $(p_K)_{K=1}^\infty$ with the properties stated in Lemma~\ref{lem:finalApproxLemma}. To do this, we first define $p'_K(x) \equiv \int_L^x p''_K(u) du + h'(L)$. Note that this satisfies the constraint $p'_K(L) = h'(L)$ and has $p''_K$ as its derivative.

We can show that $p'_K$ approximates $h'$ as follows:
\begin{align}
    ||(h' - p'_K)W_{L,4c}||_{L_p((-\infty,L])}
    &= \left( \int_{-\infty}^L |(h'(x) - p'_K(x)) e^{-4c|x-L|}|^p dx \right)^{1/p} \\
    &= \left( \int_{-\infty}^L |(h'(x) - p'_K(x))|^p e^{-p4c|x-L|} dx \right)^{1/p} \\
    &= \left( \int_{-\infty}^L |\left[\int_L^x h''(v) dv + h'(L) \right] - \left[\int_L^x p''_K(u) du + h'(L) \right]|^p e^{-p4c|x-L|} dx \right)^{1/p} \\
    &= \left( \int_{-\infty}^L |\int_L^x h''(u) - p''_K(u) du|^p e^{-p4c|x-L|} dx \right)^{1/p} \\
    &= \left( \int_{-\infty}^L |\int_x^L h''(u) - p''_K(u) du|^p e^{-p4c|x-L|} dx \right)^{1/p} \\
    &\leq \left( \int_{-\infty}^L \left[\int_x^L | h''(u) - p''_K(u)| du\right]^p e^{-p4c|x-L|} dx \right)^{1/p} \\
    &\leq \left( \int_{-\infty}^L \int_x^L | h''(u) - p''_K(u)|^p du \; e^{-p4c|x-L|} dx \right)^{1/p} \text{ (Jensen's inequality)}\\
    &= \left( \int_{-\infty}^L \int_x^L | h''(u) - p''_K(u)|^p e^{-p2c|x-L|} du \; e^{-p2c|x-L|} dx \right)^{1/p}\\
    &\leq \left( \int_{-\infty}^L \int_{x}^L | h''(u) - p''_K(u)|^p e^{-p2c|u-L|} du \; e^{-p2c|x-L|} dx \right)^{1/p}\\
    &\leq \left( \int_{-\infty}^L \int_{-\infty}^L | h''(u) - p''_K(u)|^p e^{-p2c|u-L|} du \; e^{-p2c|x-L|} dx \right)^{1/p}\\
    &= ||(h'' - p''_K)W_{L,2c}||_{L_p((-\infty,L])} \left( \int_{-\infty}^L \; e^{-p2c|x-L|} dx \right)^{1/p}
\end{align}

The above uses $p \in [1,\infty)$, but similar reasoning applies to the case of $p = \infty$. This, along with Equation~\ref{eq:approxRsltDeriv2} and the finiteness of $\int_{-\infty}^L \; e^{-p2c|x-L|} dx$ gives us our approximation result for $h'$:

\begin{align}
    \lim_{K \to \infty} ||(h' - p'_K)W_{L,2c}||_{L_p((-\infty,L])} = 0.
\end{align}

Just as we obtained $p'_K$ from $p''_K$ we integrate $p'_K$ to obtain $p_K \equiv \int_L^x p'_K(u) du + h(L)$. Again, note that $p_K(L) = h(L)$ and the derivative of $p_K$ is $p'_K$. An argument like that above provides the following approximation result

\begin{align} \label{eq:approxRsltDeriv0}
    \lim_{K \to \infty} ||(h - p_K)W_{L,c}||_{L_p((-\infty,L])} = 0.
\end{align}

Finally, note that approximation with a smaller value of $c$ also implies approximation with a larger value of $c$. Formally, if $c' > c$, then $W_{L,c'}(x) \leq W_{L,c}(x)$, so, for an arbitrary function $g$,

\begin{align}
    \lim_{K \to \infty} ||gW_{L,c}||_{L_p((-\infty,L])} = 0
    \implies \lim_{K \to \infty} ||gW_{L,c'}||_{L_p((-\infty,L])} = 0.
\end{align}

Equations~\ref{eq:approxRsltDeriv2} and \ref{eq:approxRsltDeriv0}, then, prove Lemma~\ref{lem:finalApproxLemma}.

\end{proof}

We can now prove our result on the approximation properties of polynomial extensions.

\begin{theorem}[Restatement of Theorem~\ref{thm:polyApprox}]
    Let $L \in \R$ and let $f:[L,\infty) \to \R$ be twice differentiable and a tempered distribution. Let $\mu$ be a probability measure such that the integrals $\int_L^\infty f(q) e^{(L-q)/b} d\mu(q) \leq \infty$ and $\int_L^\infty e^{(L-q)/b} d\mu(q)$ exist and are finite. With $w:(-\infty,L] \to \R$, let $\tilde{f}[w]$ denote the function
    \begin{align}
    \tilde{f}[w](q) = \left\{\begin{matrix} f(q) & q \ge L\\ w(q) & q < L\end{matrix}  \right.
    \end{align}
    
    Let $h:(-\infty,L] \to \R$ be an arbitrary twice continuously differentiable tempered distribution that satisfies $h(L) = f(L)$, $h'(L) = f'(L)$ and $h''(L) = f''(L)$. Denote the estimator $g \equiv \tilde{f}[h] - b^2 \tilde{f}[h]''$ and denote its expected squared error by 

    \begin{align}
        \alpha \equiv \int_{-\infty}^\infty \int_L^\infty \left(g(x) - f(q)\right)^2 \frac{1}{2b} e^{-|x-q|/b} d\mu(q) dx.
    \end{align}

    There exists a sequence of polynomials $(p_K)_{K=1}^\infty$ over $(\infty,L]$ that satisfy $p_K(L) = f(L)$, $p_K'(L) = f'(L)$ and $p_K''(L) = f''(L)$ such that the sequence of associated estimators $g_K \equiv \tilde{f}[p_K] - b^2 \tilde{f}[p_K]''$ satisfies

    \begin{align}
        \lim_{K \to \infty} \int_{-\infty}^\infty \int_L^\infty \left(g_K(x) - f(q)\right)^2 \frac{1}{2b} e^{-|x-q|/b} d\mu(q) dx = \alpha.
    \end{align}
\end{theorem}

\begin{proof}
    
Let $h$, $f$, $\tilde{f}$, $g_K$, and $\alpha$ be as defined in Theorem~\ref{thm:polyApprox} and let $(p_K)_{K=1}^\infty$ be a sequence of polynomials satisfying the conditions of Lemma~\ref{lem:finalApproxLemma} for the weight function $W_{L,1/2b}$.

Lemma~\ref{lem:finalApproxLemma} implies that $(p_K)_{K=1}^\infty$ retains its approximation properties for weight functions with thinner tails. To see this, first note that

$$W_{L,1/b}(x) = e^{-|x-L|/b} 
= e^{-|x-L|/2b}e^{-|x-L|/2b} = W_{L,1/2b}(x)W_{L,1/2b}(x).$$ 

In turn, we have
\begin{align} 
    ||(p_K - h) W_{L,1/b}||_{L_p((-\infty,L))}
    &= ||(p_K - h) W_{L,1/2b}W_{L,1/2b}||_{L_p((-\infty,L))} \\
    &\leq ||(p_K - h) W_{L,1/2b}||_{L_p((-\infty,L))} ||W_{L,1/2b}||_{L_p((-\infty,L))}.
\end{align}

By Lemma~\ref{lem:finalApproxLemma}, the final expression limits to 0. From this and the same line of reasoning applied to the approximation error of $p_K''$ for $h''$, we have, for $p \in [1,\infty)$,

\begin{align} \label{eq:ApproxThinnerTailsH}
        \lim_{K \to \infty} ||(p_K - h) W_{L,1/2b}||_{L_p((-\infty,L))} = \lim_{K \to \infty}||(p''_K(x) -h'')W_{L,1/2b}||_{L_p((-\infty,L))} = 0.
\end{align}

We now use Lemma~\ref{lem:finalApproxLemma} to bound the approximation error of the unbiased estimator based on the polynomial extension. Using the Minkowski inequality and Lemma~\ref{lem:finalApproxLemma}, we can bound the approximation error of the polynomial-based estimator over the range where the polynomial extension is used, $(-\infty,L]$. Recall that, for $x \in (-\infty,L]$, we have $g(x) = h(x) - b^2h''(x)$ and $g_K(x) = p_K(x) - b^2 p_K(x)''$. We obtain

\begin{align}
    ||(g_K - g)W_{L,1/b}||_{L_p((-\infty,L))} &= 
    ||[(p_K - h) - b^2(p_K'' - h'')]W_{L,1/b}||_{L_p((-\infty,L))} \\
    &\leq ||(p_K - h)W_{L,1/b}||_{L_p((-\infty,L))} + ||b^2(p_K'' - h'')W_{L,1/b}||_{L_p((-\infty,L))}
\end{align}

Applying Lemma~\ref{lem:finalApproxLemma} and using a similar line of reasoning as used above to derive Equation~\ref{eq:ApproxThinnerTailsH}, the following holds for $p \in [1,\infty)$.

\begin{align} \label{eq:ApproxThinnerTailsEst}
    \lim_{K \to \infty} ||(g_K - g)W_{L,1/b}||_{L_p((-\infty,L))} = \lim_{K \to \infty} ||(g_K - g)W_{L,1/2b}||_{L_p((-\infty,L))} = 0.
\end{align}

Recall that, for $x \geq L$, $g_K(x) = g(x)$. The expected squared error of the estimator $g_K$ can be decomposed as follows:

\begin{align}
     &\int_{-\infty}^\infty \int_L^\infty \left(g_K(x) - f(q)\right)^2 \frac{1}{2b} e^{-|x-q|/b} d\mu(q) dx \\
     &= \int_{-\infty}^L \int_L^\infty \left(g_K(x) - f(q)\right)^2 \frac{1}{2b} e^{-|x-q|/b} d\mu(q) dx + \int_{L}^\infty \int_L^\infty \left(g_K(x) - f(q)\right)^2 \frac{1}{2b} e^{-|x-q|/b} d\mu(q) dx \\
     &= \int_{-\infty}^L \int_L^\infty \left(g(x) + [g_K(x) - g(x)] - f(q)\right)^2 \frac{1}{2b} e^{-|x-q|/b} d\mu(q) dx \nonumber \\
     &+ \int_{L}^\infty \int_L^\infty \left(g(x) - f(q)\right)^2 \frac{1}{2b} e^{-|x-q|/b} d\mu(q) dx \\
     &= \int_{-\infty}^L \int_L^\infty \left[\left(g(x) - f(q)\right)^2 + 2(g(x) - f(q))(g_K(x) - g(x)) + (g_K(x) - g(x))^2 \right] \frac{1}{2b} e^{-|x-q|/b} d\mu(q) dx \nonumber \\
     &+ \int_{L}^\infty \int_L^\infty \left(g(x) - f(q)\right)^2 \frac{1}{2b} e^{-|x-q|/b} d\mu(q) dx \\
     &= \int_{-\infty}^\infty \int_L^\infty \left(g(x) - f(q)\right)^2 \frac{1}{2b} e^{-|x-q|/b} d\mu(q) dx \nonumber \\
     &+\int_{-\infty}^L \int_L^\infty \left[2(g(x) - f(q))(g_K(x) - g(x)) + (g_K(x) - g(x))^2 \right] \frac{1}{2b} e^{-|x-q|/b} d\mu(q) dx \\
     &= \alpha + \int_{-\infty}^L \int_L^\infty \left[2(g(x) - f(q))(g_K(x) - g(x)) + (g_K(x) - g(x))^2 \right] \frac{1}{2b} e^{-|x-q|/b} d\mu(q) dx \\
     &= \alpha + \int_L^\infty \int_{-\infty}^L\left[2(g(x) - f(q))(g_K(x) - g(x)) + (g_K(x) - g(x))^2 \right] \frac{1}{2b} e^{(x-q)/b} dx d\mu(q) \\
     &= \alpha + \int_L^\infty \int_{-\infty}^L\left[2(g(x) - f(q))(g_K(x) - g(x)) + (g_K(x) - g(x))^2 \right] \frac{1}{2b} e^{(x-L)/b} dx \;\; e^{(L-q)/b} d\mu(q) \\
    &= \alpha - \frac{1}{b}\int_{-\infty}^L \left[g_K(x) - g(x) \right]
    e^{(x-L)/b} dx \int_L^\infty f(q) e^{(L-q)/b} d\mu(q) \nonumber \\
    &+\frac{1}{b}\int_{-\infty}^L \left[g(x)(g_K(x) - g(x)) \right]
    e^{(x-L)/b} dx \int_L^\infty e^{(L-q)/b} d\mu(q) \nonumber \\
    &+ \frac{1}{2b} \int_{-\infty}^L \left[g_K(x) - g(x) \right]^2
    e^{(x-L)/b} dx \int_L^\infty e^{(L-q)/b} d\mu(q).\label{eq:approxErrFinalDecomp}
\end{align}

Recall that, by assumption, the integrals $\int_L^\infty e^{(L-q)/b} d\mu(q)$ and $\int_L^\infty f(q) e^{(L-q)/b} d\mu(q)$ are finite. We now need only to show that the three integrals over $x$ that are multiplied by these finite terms tend towards zero as $K \to \infty$.

For the first one, 
\begin{align}\label{eq:}
    |\int_{-\infty}^L \left[g_K(x) - g(x) \right]e^{(x-L)/b} dx| 
    &\leq \int_{-\infty}^L |\left[g_K(x) - g(x) \right]e^{(x-L)/b}| dx \\
    &= ||(g_K - g)W_{L,b}||_{L_1((-\infty,L))}.
\end{align}

By Equation~\ref{eq:ApproxThinnerTailsEst} we have $\lim_{K \to \infty} ||(g_K - g)W_{L,b}||_{L_1((-\infty,L))} = 0$, which, together with the above, implies

\begin{align} \label{eq:ApproxErrDecompLim1}
     \lim_{K \to \infty} \int_{-\infty}^L \left[g_K(x) - g(x) \right]e^{(x-L)/b} dx = 0.
\end{align}

For the second integral, we apply H\"older's inequality:
\begin{align}
    |\int_{-\infty}^L \left[g(x)(g_K(x) - g(x)) \right] e^{(x-L)/b} dx|
    &\leq \int_{-\infty}^L |\left[g(x)(g_K(x) - g(x)) \right] e^{(x-L)/2b}e^{(x-L)/2b}| dx \\
    &= ||\left[g(x)(g_K(x) - g(x)) \right]W_{L,1/2b}W_{L,1/2b}||_{L_1((-\infty,L))} \\
    &\leq ||g(x)W_{L,1/2b}||_{L_\infty((-\infty,L))} \\
    &\;\;\;\;\cdot||\left[g_K(x) - g(x) \right]W_{L,1/2b}||_{L_1((-\infty,L))}
\end{align}

$||g(x)W_{L,1/2b}||_{L_\infty((-\infty,L))}$ is finite because $g$ is a tempered distribution. This, combined with Equation~\ref{eq:ApproxThinnerTailsEst} and the above implies that

\begin{align} \label{eq:ApproxErrDecompLim2}
    \lim_{K \to \infty} \int_{-\infty}^L \left[g(x)(g_K(x) - g(x)) \right] e^{(x-L)/b} dx = 0.
\end{align}

Finally, for the third integral, we have

\begin{align}
    |\int_{-\infty}^L \left[g_K(x) - g(x)\right]^2 e^{(x-L)/b} dx|
    &= \int_{-\infty}^L \left[(g_K(x) - g(x))
    \frac{1}{2b} e^{(x-L)/2b}\right]^2 dx \\
    &= ||(g_K - g)W_{L,1/2b}||_{L_2((-\infty,L))}^2.
\end{align}

Again, Equation~\ref{eq:ApproxThinnerTailsEst} states that the final term limits to zero, and so

\begin{align} \label{eq:ApproxErrDecompLim3}
    \lim_{K \to \infty} \int_{-\infty}^L \left[g_K(x) - g(x)\right]^2 e^{(x-L)/b} dx = 0.
\end{align}

Together, Equations~\ref{eq:ApproxErrDecompLim1}, \ref{eq:ApproxErrDecompLim2}, \ref{eq:ApproxErrDecompLim3}, and \ref{eq:approxErrFinalDecomp} imply that the expected squared error of $g_K$ limits to $\alpha$, proving Theorem~\ref{thm:polyApprox}.

\end{proof}

\subsection{Proof of Theorem 15}
\label{app:thm15}

\begin{theorem}
    (Restatement of Theorem~\ref{thm:optimal}) For any positive integer $k$, any real number $L \in \RR$, and any function $f$ which is twice differentiable on $[L,\infty)$, there is an algorithm that runs in time $poly(k)$ which computes the polynomial $g$ that minimizes 
    $$\int_{-\infty}^L \int_L^\infty \left(g(x)-f(q)\right)^2 \frac{1}{2b} e^{(x-q)/b} d\mu(q) dx$$
    over polynomials of degree $k$, satisfying the constraints $g(x) = h(x) - b^2 h''(x)$, $h(L) = f(L), h'(L) = f'(L)$, and $h''(L) = f''(L)$.
\end{theorem}

\begin{proof}
To construct this algorithm, we first show that our objective function is convex for any finite or bounded measure $\mu$:

The value of our objective function, in terms of the coefficient of the polynomial $a_i$, is
\begin{align}
&\int_{-\infty}^L \int_L^\infty (g(x)-f(q))^2 \frac{1}{2b} e^{(x-q)/b} d\mu(q) dx \nonumber \\ 
&=  \int_{-\infty}^L \int_L^\infty \left(\sum_{i=0}^k a^2_i x^{2i} + \sum_{i=0}^{k-1} \sum_{j=i+1}^k 2a_ia_j x^{i+j}  - \sum_{i=0}^k 2f(q) a_i x^i + f(q)^2\right) \frac{1}{2b} e^{(x-q)/b} d\mu(q)dx\\ 
&= \sum_{i=0}^k a_i^2 \int_{-\infty}^L \int_L^\infty \frac{x^{2i}}{2b} e^{(x-q)/b} \mu(q) dq dx + \sum_{i=0}^{k-1} \sum_{j=i+1}^k 2a_ia_j \int_{-\infty}^L \int_L^\infty \frac{x^{i+j}}{2b} e^{(x-q)/b} d\mu(q)dx \\ 
&\quad- \sum_{i=0}^k 2a_i \int_{-\infty}^L \int_L^\infty \frac{x^if(q)}{2b} e^{(x-q)/b} \mu(q) dq dx + \int_{-\infty}^L \int_L^\infty \frac{f(q)^2}{2b}e^{(x-q)/b} d\mu(q)dx. \nonumber
\end{align}

So the second derivative with respect to $a_i$ is constant with respect to the values of $a_i$: $$2\int_{-\infty}^L \int_L^\infty \frac{x^{2i}}{2b} e^{(x-q)/b} d\mu(q) dx$$
and the other terms in the Hessian not on the diagonal are also constant with respect to $a_i$:
$$2\int_{-\infty}^L \int_L^\infty \frac{x^{i+j}}{2b} e^{(x-q)/b} d\mu(q) dx.$$

Now, to show that the function is convex, we must show that the Hessian matrix is positive semi-definite and give conditions for the convergence of the integral. Denote the Hessian matrix by $H$ and consider some $x = (x_0,x_1, x_2, \dots, x_n)$:
\begin{align}
    x^T H x = \sum_{i=0}^k \sum_{j=0}^k x_ix_j H_{ij}.
\end{align}

Now, note that $x_i,x_j$ are constants and can be brought into the integral in the expression for $H_{ij}$. Thus, if we can show that
\begin{align}
    x_ix_j \frac{x^{i+j}}{b}e^{(x-q)/b} \mu(q) \ge 0
\end{align}
for all $x_i,x_j$, the sum of nonnegative numbers is nonnegative, so we would have $x^T H x \ge 0$ for all $x$.

Now, note that if $\mu$ is a nonnegative measure, for fixed $x,q$, we have $\frac{1}{b} e^{(x-q)/b} \mu(q) \ge 0$, so it suffices to show that $x_i x_j x^{i+j} \ge 0$. Equivalently, we just need to show that the matrix $H'$ where $H'_{ij} = x^{i+j}$ is positive semidefinite.

To see this, note that $H'_i = x^i \begin{bmatrix} x & x^2 & \dots & x^k \end{bmatrix}$. Thus, $H'$ is a rank $1$ matrix with eigenvector $\begin{bmatrix} x & x^2 & \dots & x^k \end{bmatrix}$ and eigenvalue $x^2+x^4+ \dots + x^{2k} \ge 0$. Thus, $H'$ is PSD, so $H$ is PSD as well.

Furthermore, we claim that this integral converges for all probability measures $\mu$. To see this, note that
\begin{align}
    2\int_{-\infty}^L \int_L^\infty \frac{x^{2i}}{2b} e^{(x-q)/b} \mu(q) dq dx \le \max_{q \in [L, \infty)} 2\int_{-\infty}^L \frac{x^{2i}}{2b} e^{(x-q)/b} dx = e^{-L} \cdot 2\int_{-\infty}^L \frac{x^{2i}}{2b} e^{x/b} dx,
\end{align}
and this converges, since the integrand goes to $0$ at a superpolynomial rate. Alternatively, if $\mu$ is a bounded metric, we can similarly show convergence. That is, suppose $\mu(q) \le \alpha$, 
\begin{align}
    2\int_{-\infty}^L \int_L^\infty \frac{x^{2i}}{2b} e^{(x-q)/b} \mu(q) dq dx \le \alpha \cdot \max_{q \in [L, \infty)} 2\int_{-\infty}^L \frac{x^{2i}}{2b} e^{(x-q)/b} dx = \alpha e^{-L} 2\int_{-\infty}^L \frac{x^{2i}}{2b} e^{x/b} dx,
\end{align}
and this also converges, since it only differs by a constant factor from the above. The other integrals can similarly be shown to converge for any tempered distribution $f$ as it grows at most polynomially in the tails. Thus, the objective function is convex as long as our measure is finite or bounded.

Then, the constraints are:
\begin{align}
    h(L) = f(L), h'(L) = f'(L), h''(L) = f''(L).
\end{align}

$f(L), f'(L)$, and $f''(L)$ are all constants with respect to $a_i$, and each of $h(L), h'(L),$ and $h''(L)$ are linear in $a_i$. Thus, these are linear equality constraints.

Thus, we have a convex objective function and linear equality constraints. Standard methods of optimization only require oracle access to the objective function, its gradient, and sometimes its Hessian, e.g. using one of the algorithms in chapter 10 of \cite{Boyd_Vandenberghe_2004}. Thus, this is sufficient to perform optimization. This also does not depend at all on how we integrate with respect to $q$, other than requiring convergence.
\end{proof}

\subsection{Optimization Details}\label{app:optimization}
While we showed in the main body of the paper that the optimization is possible, integrals can be slow to compute numerically. Thus, we give a more efficient method to implement the optimization. To do this, note that the objective function,

\begin{align*}
&\sum_{i=0}^k a_i^2 \int_{-\infty}^L \int_L^\infty \frac{x^{2i}}{2b} e^{(x-q)/b} d\mu(q) dx + \sum_{i=0}^{k-1} \sum_{j=i+1}^k 2a_ia_j \int_{-\infty}^L \int_L^\infty \frac{x^{i+j}}{2b} e^{(x-q)/b} d\mu(q) dx \\
&\quad- \sum_{i=0}^k 2a_i \int_{-\infty}^L \int_L^\infty \frac{x^if(q)}{2b} e^{(x-q)/b} d\mu(q) dx + \int_{-\infty}^L \int_L^\infty \frac{f(q)^2}{2b}e^{(x-q)/b} d\mu(q)dx
\end{align*}
is a quadratic in the coefficients of the polynomial $a_i$. Furthermore, since our only constraints are linear equalities, we can convert this into an unconstrained optimization problem in $k-2$ variables by isolating $a_0, a_1,$ and $a_2$ in the constraints and substituting. That is, the constraint $h(L) = f(L)$ corresponds to
\begin{align}
    \sum_{i=0}^k b_i L^i = f(L)
\end{align}
and since $h(L) - b^2 h''(L) = g(L)$, we have

\begin{align}
    a_i = b_i - b^2 i(i-1) b_{i+2}
\end{align}
and $a_k = b_k, a_{k-1} = b_{k-1}$.

We can rewrite the above as
\begin{align}
    b_i = a_i + b^2 i(i-1) b_{i+2}
\end{align}
and applying this recursively upwards on $b_{i+2}$, we get
\begin{align}
    b_i = a_i + \sum_{l=1}^{\lfloor (k-i)/2\rfloor} \frac{(i+2l)!}{i!} a_{i+2l}.
\end{align}

We can now plug this into the equality constraints to solve for $a_0, a_1,$ and $a_2$ in terms of the other coefficients. When doing so, we will get $a_0$ in terms of $a_1, a_2, \dots, a_k$ and $a_1$ in terms of $a_2, a_3, \dots, a_k$, but these can be converted to be purely in terms of $a_3, \dots, a_k$ by plugging in the expression for $a_2$ to get $a_1$ in terms of $a_3, \dots, a_k$, and then doing this again with both $a_1$ and $a_2$ to get $a_0$ in terms of $a_3, \dots, a_k$.

By doing this, we get the following equations:
\begin{align}
    f''(L) = \sum_{i=2}^k b_i L^{i-2} i(i-1) = \sum_{i=2}^k \left(a_i + \sum_{l=1}^{\lfloor (k-i)/2\rfloor} \frac{(i+2l)!}{i!} a_{i+2l}\right) L^{i-2} i(i-1)
\end{align}
\begin{align}
    a_2 = \frac{1}{2}\left(f''(L) - \sum_{i=3}^k a_i L^{i-2}i(i-1) - \sum_{i=2}^k \sum_{l=1}^{\lfloor (k-i)/2\rfloor} \frac{(i+2l)!}{i!} a_{i+2l}L^{i-2}i(i-1)\right)
\end{align}
\begin{align}
    f'(L) = \sum_{i=1}^k b_i L^{i-1} i = \sum_{i=1}^k \left(a_i + \sum_{l=1}^{\lfloor (k-i)/2\rfloor} \frac{(i+2l)!}{i!} a_{i+2l}\right) L^{i-1}i
\end{align}
\begin{align}
    a_1 = f'(L) - \sum_{i=2}^k a_i L^{i-1}i - \sum_{i=1}^k\sum_{l=1}^{\lfloor (k-i)/2\rfloor} \frac{(i+2l)!}{i!} a_{i+2l}L^{i-1}i
\end{align}
\begin{align}
    f(L) = \sum_{i=0}^k b_i L^i = \sum_{i=0}^k L^i \left( a_i + \sum_{l=1}^{\lfloor (k-i)/2\rfloor} \frac{(i+2l)!}{i!} a_{i+2l}\right)
\end{align}
\begin{align}
    a_0 = f(L) - \sum_{i=1}^k a_i L^i - \sum_{i=0}^k \sum_{l=1}^{\lfloor (k-i)/2\rfloor} \frac{(i+2l)!}{i!} a_{i+2l}L^i.
\end{align}

Now, to perform our optimization, we do the following:

\begin{enumerate}
    \item Our variables are $a_3, \dots, a_k$.
    \item We initialize the state as $0^{k-2}$, corresponding to the $2^{nd}$ degree Taylor expansion of $f$ at $L$.
    \item We compute the Hessian and Gradient of our objective function after plugging in the above values for $a_0, a_1$, and $a_2$.
    \item We perform 1 Newton step to compute the optimal values of $a_3, \dots, a_k$.
    \item We plug those values into the above formulas to compute the values of $a_0, a_1, a_2$.
\end{enumerate}

To compute the gradient and hessian, we can consider each additive term in the objective function individually. Everything in the objective function is constant with respect to $a_i$ except for terms with $a_0, a_1, a_2$, and $a_i$ in them. Thus, to compute the gradient, we compute:
\begin{align}
    \frac{\partial a_2}{\partial a_i} = -L^{i-2} i(i-1) - \sum_{j=1}^{\lfloor i/2\rfloor-1} \frac{i!}{(i-2j)!} L^{i-2j-2} (i-2j)(i-2j-1)
\end{align}
\begin{align}
    \frac{\partial a_1}{\partial a_i} = -2L\frac{\partial a_2}{\partial a_i} - L^{i-1} i - \sum_{j=1}^{\lfloor i/2-1\rfloor} \frac{i!}{(i-2j)!} L^{i-2j-1} (i-2j)
\end{align}
\begin{align}
    \frac{\partial a_0}{\partial a_i} = -(2 + L^2)\frac{\partial a_2}{\partial a_i} -L\frac{\partial a_1}{\partial a_i} - L^i - \sum_{j=1}^{\lfloor i/2\rfloor} \frac{i!}{(i-2j)!} L^{i-2j}.
\end{align}

Now, when computing the gradient, we simply compute it termwise. A term will appear in the $G_i$ iff it contains $a_i, a_2, a_1,$ or $a_0$. Furthermore, since we initialize at $a_3 = a_4 = \dots = a_k = 0$, we can safely ignore any term whose derivative contains $a_j$ for $j \ge 3$. If we let the objective function be $\varphi$ and $G$ be the gradient, we end up getting

{
\tiny
\begin{equation}\label{eq:gradient}G_i = \frac{1}{2b} \int_{-\infty}^L \int_L^\infty e^{(x-q)/b} \left\{2f(q) x^i + \sum_{\alpha=0}^2 \left[2a_\alpha x^{i+\alpha} + 2f(q)x^\alpha \frac{\partial a_\alpha}{\partial a_i} + \sum_{\beta=\alpha}^2 \left(2\frac{\partial a_\alpha}{\partial a_i} a_\beta x^{\alpha+\beta} + 2\frac{\partial a_\beta}{\partial a_i} a_\alpha x^{\alpha+\beta}\right)\right]\right\} d\mu(q) dx.
\end{equation}
}

To compute the Hessian, we can similarly compute it termwise, where a term will appear in $H_{ij}$ iff it contains any $a_{i'}a_{j'}$ for $i' \in \{i,2,1,0\}$ and $j' \in \{j,2,1,0\}$. This gets us:

{
\tiny
\begin{equation}\label{eq:hessian}H_{ij} = \frac{1}{2b}  \int_{-\infty}^L \int_L^\infty e^{(x-q)/b}\left\{ 2x^{i+j} + \sum_{\alpha=0}^2 \left[2\frac{\partial a_\alpha}{\partial a_i}x^{\alpha+j} + 2\frac{\partial a_\alpha}{\partial a_j}x^{\alpha+i} + \sum_{\beta=\alpha}^2\left( 2\frac{\partial a_\alpha}{\partial a_i}\frac{\partial a_\beta}{\partial a_j}x^{\alpha+\beta} + 2\frac{\partial a_\alpha}{\partial a_j}\frac{\partial a_\beta}{\partial a_i}x^{\alpha+\beta}\right)\right]\right\} d\mu(q) dx. \end{equation}
}

These expressions were derived as follows:

For the gradient, the nonzero terms in $G_i$ are the partial derivatives of terms with $a_i, a_{\alpha}a_i, a_\alpha,$ and $a_\alpha a_\beta$ for $0 \le \alpha,\beta \le 2$.

Furthermore, since the gradient is linear, we can simply take all of these partial derivatives separately and add them together.

Now, recall the value of the objective function:
\begin{align*}
&\sum_{i=0}^k a_i^2 \int_{-\infty}^L \int_L^\infty \frac{x^{2i}}{2b} e^{(x-q)/b} d\mu(q) dx + \sum_{i=0}^{k-1} \sum_{j=i+1}^k 2a_ia_j \int_{-\infty}^L \int_L^\infty \frac{x^{i+j}}{2b} e^{(x-q)/b} d\mu(q) dx \\
&\quad- \sum_{i=0}^k 2a_i \int_{-\infty}^L \int_L^\infty \frac{x^if(q)}{2b} e^{(x-q)/b} d\mu(q) dx + \int_{-\infty}^L \int_L^\infty \frac{f(q)^2}{2b}e^{(x-q)/b} d\mu(q)dx
\end{align*}

The term with $a_i$ adds
$$2\int_{-\infty}^L \int_L^\infty \frac{x^if(q)}{2b} e^{(x-q)/b} d\mu(q) dx$$

The terms with $a_i a_\alpha$ for $0 \le \alpha \le 2$ add
$$2 a_\alpha\int_{-\infty}^L \int_L^\infty \frac{x^{i+\alpha}}{2b} e^{(x-q)/b} d\mu(q) dx$$

The term with $a_\alpha$ for $0 \le \alpha \le 2$ adds
$$2\frac{\partial a_\alpha}{\partial a_i}\int_{-\infty}^L \int_L^\infty \frac{x^\alpha f(q)}{2b} e^{(x-q)/b} d\mu(q) dx$$

and finally, the terms with $a_\alpha a_\beta$ for $0\le \beta \le \alpha \le 2$ add

$$2 a_\alpha \frac{\partial a_\beta}{a_i}\int_{-\infty}^L \int_L^\infty \frac{x^{\beta+\alpha}}{2b} e^{(x-q)/b} d\mu(q) dx + 2 a_\beta \frac{\partial a_\alpha}{a_i}\int_{-\infty}^L \int_L^\infty \frac{x^{\beta+\alpha}}{2b} e^{(x-q)/b} d\mu(q) dx$$

Adding all these together gets us the expression for $G_i$.

We can do the same thing for the Hessian. This time, the only terms that can contribute are the ones with $a_ia_j$, $a_i a_\alpha$, $a_j a_\alpha$, and $a_\alpha a_\beta$.

The terms with $a_ia_j$ add
$$2 \int_{-\infty}^L \int_L^\infty \frac{x^{i+j}}{2b} e^{(x-q)/b} d\mu(q) dx$$

The terms with $a_i a_\alpha$ add
$$2 \frac{\partial a_\alpha}{a_j}\int_{-\infty}^L \int_L^\infty \frac{x^{i+\alpha}}{2b} e^{(x-q)/b} d\mu(q) dx$$

The terms with $a_j a_\alpha$ add
$$2 \frac{\partial a_\alpha}{a_i}\int_{-\infty}^L \int_L^\infty \frac{x^{i+\alpha}}{2b} e^{(x-q)/b} d\mu(q) dx$$

and finally, the terms with $a_\alpha a_\beta$ add

$$2 \frac{\partial a_\alpha}{a_j} \frac{\partial a_\beta}{a_i}\int_{-\infty}^L \int_L^\infty \frac{x^{\beta+\alpha}}{2b} e^{(x-q)/b} d\mu(q) dx + 2 \frac{\partial a_\beta}{a_j} \frac{\partial a_\alpha}{a_i}\int_{-\infty}^L \int_L^\infty \frac{x^{\beta+\alpha}}{2b} e^{(x-q)/b} d\mu(q) dx$$

Again, adding all of these up gets us the expression for $H_i$.

\section{Proofs for Section~\ref{sec:PRDP}}
\label{app:PRDPProofs}

\begin{theorem}(Restatement of Theorem~\ref{thm:transformPRDP})
    Assume the query value $q(D) \in [0,\infty)$; the offset parameter $a \in \R$; the noise scale parameter $b \in (0,\infty)$; the transformation function $f:[a,\infty) \to \mathcal{F} \subseteq \R$ is concave and strictly increasing; and the estimator $g:\mathcal{F} \to \mathcal{G} \subseteq \R$. Denote by $\Delta_f(r)$ the per-record sensitivity of the query $f(q(D) + a)$, as defined in Definition~\ref{def:PR-sensitivity}. The output of Algorithm~$\ref{alg:transform}\allowbreak(q(D), a, b, f, g)$ satisfies $P$-PRDP for $P(r) = \frac{\Delta_f(r)}{b}$.
\end{theorem}

\paragraph{Proof of Theorem~\ref{thm:transformPRDP}}

Let $d_\alpha (X \| Y)$ denote the R\'enyi divergence of order $\alpha$ between the distributions of the random variables $X$ and $Y$. An equivalent formulation of Definition~\ref{def:PRDP} (Definition 5.1 in \cite{finleyEtAl2024}) states that a mechanism $M$ satisfies PRDP for any policy function $P$ that, for all $r$, satisfies

$$\sup_{D,D^\prime \text{ such that } D \ominus D^\prime = \{r\}} d_\infty \left(M(D) \| M(D') \right) \leq P(r).$$
    
Let $M(D)$ denote Algorithm~$\ref{alg:transform}\allowbreak(q(D), a, b, f, g)$. Note that $M(D) = g(\Lap(f(q(D) + a), b))$. We obtain $P(r)$, by using the data processing inequality \cite{renyi1961measures}, the R\'enyi divergence between two Laplace distributions (see, for example, page 8 of \cite{Mironov2017RenyiDP}), and Definition~\ref{def:PR-sensitivity} as follows:
    \begin{align}
        &\sup_{D,D^\prime \text{ such that } D \ominus D^\prime = \{r\}} d_\infty \left(M(D) \| M(D') \right) \\
        &\leq \sup_{D,D^\prime \text{ such that } D \ominus D^\prime = \{r\}} d_\infty \left(\Lap(f(q(D) + a), b) \| \Lap(f(q(D') + a), b) \right) \\
        &= \sup_{D,D^\prime \text{ such that } D \ominus D^\prime = \{r\}} \frac{|f(q(D) + a) - f(q(D') + a)|}{b} \\
        &= \frac{\Delta_f(r)}{b} \\
        &= P(r). 
    \end{align}  

\begin{theorem}
(Restatement of Theorem~\ref{thm:privTrnsSum})
Let the assumptions of Theorem~\ref{thm:transformPRDP} hold, and further assume $a \geq 0$ and $r.c \geq 0$ for all records $r$. For the sum query $q(D) = \sum_{r \in D} r.c$, the per-record sensitivity of $f(q(D) + a)$ is $\Delta_f(r) = f(r.c + a) - f(a)$ and Algorithm~\ref{alg:transform}($q(D),a,b,f,g$) satisfies PRDP with the policy function $P(r) = [f(r.c + a) - f(a)]/b$.
\end{theorem}

\paragraph{Proof of Theorem~\ref{thm:privTrnsSum}}

By the symmetry of the neighbor relationship, we assume WLOG that $D' = D \cup \{r\}$.
    \begin{align}
        \Delta_f(r) &\equiv 
        \sup_{D,D^\prime \text{ such that } D \ominus D^\prime = \{r\}} |f(a + \sum_{s \in D} s.c) - f(a + \sum_{s \in D'} s.c)| \\
        &= \sup_{D; r} |f(a + \sum_{s \in D} s.c) - f(a + r.c + \sum_{s \in D} s.c) |.
    \end{align}

    $f$ is concave and increasing, so the final supremand is decreasing in $\sum_{s \in D} s.c$, which is bounded below by 0. Therefore,
    $$\Delta_f(r) = f(r.c + a) - f(a).$$ 

    It follows from Theorem~\ref{thm:transformPRDP}, then, that this mechanism has a policy function of $P(r) = [f(r.c + a) - f(a)]/b$.

\bibliographystyle{plain}
\bibliography{ref}

\end{document}